\documentclass{article}
\usepackage[utf8]{inputenc}
\usepackage[english]{babel}
\usepackage{amsmath,amsfonts,amssymb,amsthm}
\usepackage{geometry}
\usepackage{tikz}
\usepackage{standalone}

\usepackage{graphicx}
\usepackage{caption}    	
\usepackage{subcaption}

\usepackage{algorithm}
\usepackage{algorithmic}
\floatname{algorithm}{Algorithm}
\usepackage{color}      		

\definecolor{gray97}{gray}{.97}
\definecolor{gray75}{gray}{.45}
\definecolor{gray90}{gray}{.85}

\theoremstyle{plain}
\newtheorem{prop}{Proposition}[section]
\newtheorem{teor}{Theorem}[section]

\theoremstyle{definition}

\newcommand{\espacio}{\text{\textvisiblespace}}
\newcommand{\NN}{\mathbb{N}}

\newcommand{\Z}{\mathbb{Z}}
\newcommand{\A}{\mathcal A}
\renewcommand{\phi}{\varphi}

\title{Nontrivial Turmites are Turing-universal}
\author{Diego Maldonado$^\dag$\\
LIFO, Universit\'e d'Orl\'eans, France\and\\
Anah\'i Gajardo$^\dag$\footnote{The author thanks the CMM and CI$^2$MA CONICYT Basal project PFB03.}\\
DIM \& CI$^2$MA, Universidad de Concepci\'on, Chile\and\\
Benjamin Hellouin de Menibus\footnote{The author acknowledges the financial support of Basal project No. PFB-03 CMM, Universidad de Chile.}\\
Center for Mathematical Modeling, Universidad de Chile, Chile\\
Departamento de Matem\'aticas, Universidad Andres Bello, Chile\and\\
Andr\'es Moreira\footnote{The authors thank the support of the FONDECYT project\#1140833 and the ECOS-Sud project \#C12E05.}\ $^,$ \footnote{\tt Corresponding author:amoreira@inf.utfsm.cl}\\
Departamento de Inform\'atica, Universidad T\'ecnica Federico Santa Mar\'ia, Chile}
\begin{document}

\maketitle
\begin{abstract}
A Turmit is a Turing machine that works over a two-dimensional grid, that is, an agent that moves, reads and writes symbols over the cells of the grid. 
Its state is an arrow and, depending on the symbol that it reads, it turns to the left or to the right, switching the symbol at the same time.
Several symbols are admitted, and the rule is specified by the turning sense that the machine has over each symbol.
Turmites are a generalization of Langton's ant, and they present very complex and diverse behaviors.
We prove that any Turmite, except for those whose rule does not depend on the symbol, can simulate any Turing Machine.
We also prove the \textbf{P}-completeness of prediction their future behavior by explicitly giving a log-space reduction from the \emph{Topological Circuit Value Problem}.
A similar result was already established for Langton's ant; here we use a similar technique but prove a stronger notion of simulation, and for a more general family.
\end{abstract}
\emph{(Turing-universality, P-completeness, One Head Machines, Unconventional Computation)}


Langton's ant \cite{Langton86} is a Turing machine evolving on a two-dimensional grid where each vertex contains a white or black symbol. 
The ``ant" is the head of the machine and is in one state corresponding to its current direction (N, S, E, O). 
At each time step, the ant moves one step forward and turns according to the symbol on the new vertex, changing this symbol.

Despite the simplicity of the rules governing its dynamics, the ant presents extremely complex trajectories. 
From a white configuration, the ant initially exhibits a simple, almost symmetric behavior, but then passes through the path it has drawn in an apparently disordered, chaotic fashion, before entering in a recurrent behavior that takes it away from the initial position. 
It remains an open problem if this periodic asymptotic behavior is the same for any initial configuration with finitely many black symbols.

Various generalizations of this model have been suggested: several ants \cite{Langton86}, other regular grids or finite grids \cite{GGM01}, more than two symbols \cite{GalProSutTro95a.mi}. 
Among the latter are the Turmites, first discussed in~\cite{Dewd89} and credited to Greg Turk. Here the symbols are taken from $\Z/n\Z$ and each one corresponds to turning right (R) or left (L). 
When the Turmite visits a cell it turns according to the symbol in the cell and add one to this. 
These Turmites preserve several properties of Langton's ant, such as reversibility and unboundness, but they exhibit much diverse and complex asymptotic behaviors.

In this paper we prove that, with the exception of two rules with trivial dynamics, all Turmites can perform universal computation in the sense used in~\cite{Cook04a} for the Rule 110 cellular automaton: that is, we prove that given a 1D Turing machine, there exists a periodic configuration of the Turmite space which can be perturbed by a finite pattern, dependent on any given finite input word, such that the Turmite simulates the Turing machine over that word. 

With a similar technique, we prove that the problem of deciding if a given cell is visited by the trajectory of the Turmite before some time is \textbf{P}-complete. \textbf{P}-completeness is regularly used in similar contexts to show that the problem of predicting the future behavior of a system is intrinsically hard to parallelize (under unproven but widely held assumptions in complexity theory). An exposition of these concepts can be found in \cite{Greenlaw95}. \textbf{P}-completeness results have be proven in this way in, for example, the Rule 110 cellular \cite{NearyWoods}, two-dimensional cellular automata such as the Game of Life~\cite{DurRok99} and the original Langton's ant \cite{GGM01}. 

The fundamental idea of this proof was introduced by Banks \cite{AITR-233} and consists in the simulation of logical circuits. A similar technique was used for two-dimensional cellular automata (see \cite{GGchaitin} for a review on the subject). Turmites (and Langton's ant) work in a slightly different way than these previous examples, because it is the Turmite that needs to go over the cellular space to perform the computation of each logical gate.
In this work, we use the ideas from~\cite{GGM01} but creating very particular gadgets that allow a construction that works for all the Turmites at the same time.
Moreover, the infinite circuit that simulates a Turing machine was modified such that now it is completely periodic except for a finite number of cells.
The proof of {\bf P}-completeness is done by reduction from the \emph{Topological circuit value problem} with a detailed and explicit construction, that is proved to be computable in logarithmic space.

The formal definition of a Turmite is presented in Section~\ref{sec:def}.
Section~\ref{sec: sistema abstracto} is devoted to the definition of a set of computational devices from which one can simulate any Turing machine.
We prove that these devices can be implemented by any Turmite in Section~\ref{sec:Turmit}, establishing their universality.
Finally, in Section~\ref{sec:Phardness} the proof of P-completeness is developed.

\section{Langton's ant and Turmites}\label{sec:def}

As defined in~\cite{GalProSutTro95a.mi}, a Turmite is an automaton with state set $Q=\{\rightarrow,\uparrow,\leftarrow,\downarrow\}$ and alphabet $\A = \{0,1,...,n-1\}=\Z/\Z_n $, that moves over the grid $\Z^2$.
In this setting, a configuration is an element from $\A^{\Z^2}\times \Z^2\times Q$, that represents:

\begin{itemize}
\item an assignment of symbols from $\A$ to each cell in $\Z^2$,
\item a marked cell representing the position of the Turmite, and
\item the current state of the Turmite.
\end{itemize}

At each cell the Turmite turns either to the right or to the left, depending on its \emph{rule} and the symbol of the cell where it lands, it augments the symbol by 1 and then it moves one cell forward.
Let us note that the system is invariant under rotations of the configuration and the Turmite state, but not under reflections.

A particular Turmite rule is characterized by a word $w\in\{L,R\}^n$ that assigns to each symbol a turning direction.
Formally, for a given configuration $(C,(i,j),d)\in\A^{\Z^2}\times \Z^2\times Q$, the global transition function of the rule $w$ is defined by $T_w(C,(i,j),d)=(C',(i',j'),d')$, where:
\begin{itemize}
\item $C'(k,l)=\left\{\begin{array}{ll}
  C(k,l)+1\mod n & \text{if } (k,l)=(i,j)\\
  C(k,l)   & \text{otherwise} 
\end{array}\right.$,
\item $d'$ is a rotation of $d$ by $90^\circ$ clockwise if $w_{C(i,j)}=R$, or counterclockwise if $w_{C(i,j)}=L$,
\item $(i',j')=(i,j)+d'$.
\end{itemize}

An illustration of the dynamics of rule $RRL$ is shown in Figure~\ref{fig: ejem RRL}. 

\begin{figure}[H]
  \centering
  \begin{subfigure}[b]{0.15\textwidth}
    \centering
        {
          \begin{tikzpicture}
\draw (-4,5) node (v1) {} rectangle (1,0) node (v2) {};
\foreach \i in {-4,...,1} {
        \draw [very thin,gray] (\i,0) -- (\i,5);
    }
\foreach \i in {0,...,5} {
        \draw [very thin,gray] (-4,\i) -- (1,\i);
    }
\draw [fill, black]  (-4,2) rectangle (-3,1);
\draw [fill, black]  (-3,4) rectangle (-2,3);
\draw [fill, black] (-1,1) rectangle (0,0);
\draw [fill, gray] (-1,4) rectangle (0,3);
\draw [fill, black] (-1,3) rectangle (0,2);
\fill [red] (-2,3) -- (-1,3) -- (-1.5,4)--cycle;
\end{tikzpicture}
        }
        \caption{$t=1$}     
  \end{subfigure}
  \begin{subfigure}[b]{0.15\textwidth}
    \centering
        {
          \begin{tikzpicture}
\draw[fill, black]  (-4,2) rectangle (-3,1);
\draw[fill, black]  (-3,4) rectangle (-2,3);
\draw [fill, black] (-1,1) rectangle (0,0);
\draw [fill, gray] (-1,4) rectangle (0,3);
\draw [fill, gray] (-2,4) rectangle (-1,3);
\draw [fill, black] (-1,3) rectangle (0,2);
\fill [red] (-1,4) -- (-1,3) -- (0,3.5)--cycle;
\draw (-4,5) node (v1) {} rectangle (1,0) node (v2) {};
\foreach \i in {-4,...,1} {
        \draw [very thin,gray] (\i,0) -- (\i,5);
    }
\foreach \i in {0,...,5} {
        \draw [very thin,gray] (-4,\i) -- (1,\i);
    }
\end{tikzpicture} 
        }     
        \caption{$t=2$}
  \end{subfigure}
  \begin{subfigure}[b]{0.15\textwidth}
    \centering
        {
          \begin{tikzpicture}
\draw [fill, black] (-4,2) rectangle (-3,1);
\draw [fill, black] (-3,4) rectangle (-2,3);
\draw [fill, black] (-1,1) rectangle (0,0);
\draw [fill, black] (-1,4) rectangle (0,3);
\draw [fill, gray] (-2,4) rectangle (-1,3);
\draw [fill, black] (-1,3) rectangle (0,2);
\fill [red] (-1,3) -- (0,3) -- (-0.5,2)--cycle;
\draw (-4,5) node (v1) {} rectangle (1,0) node (v2) {};
\foreach \i in {-4,...,1} {
        \draw [very thin,gray] (\i,0) -- (\i,5);
    }
\foreach \i in {0,...,5} {
        \draw [very thin,gray] (-4,\i) -- (1,\i);
    }
\end{tikzpicture}
        }
        \caption{$t=3$}
  \end{subfigure}   
  \begin{subfigure}[b]{0.15\textwidth}
    \centering
        {
          \begin{tikzpicture}
\draw [fill, black] (-4,2) rectangle (-3,1);
\draw [fill, black] (-3,4) rectangle (-2,3);
\draw [fill, black] (-1,1) rectangle (0,0);
\draw [fill, black] (-1,4) rectangle (0,3);
\draw [fill, gray] (-2,4) rectangle (-1,3);
\fill [red] (0,3) -- (0,2) -- (1,2.5)--cycle;
\draw (-4,5) node (v1) {} rectangle (1,0) node (v2) {};
\foreach \i in {-4,...,1} {
        \draw [very thin,gray] (\i,0) -- (\i,5);
    }
\foreach \i in {0,...,5} {
        \draw [very thin,gray] (-4,\i) -- (1,\i);
    }
\end{tikzpicture}
        }
        \caption{$t=4$}
  \end{subfigure} 
  \begin{subfigure}[b]{0.15\textwidth}
    \centering
        {
          \begin{tikzpicture}
\draw [fill, black] (-4,2) rectangle (-3,1);
\draw [fill, black] (-3,4) rectangle (-2,3);
\draw [fill, black] (-1,1) rectangle (0,0);
\draw [fill, black] (-1,4) rectangle (0,3);
\draw [fill, gray] (-2,4) rectangle (-1,3);
\draw [fill, gray] (0,3) rectangle (1,2);
\fill [red] (0,2) -- (1,2) -- (0.5,1)--cycle;
\draw (-4,5) node (v1) {} rectangle (1,0) node (v2) {};
\foreach \i in {-4,...,1} {
        \draw [very thin,gray] (\i,0) -- (\i,5);
    }
\foreach \i in {0,...,5} {
        \draw [very thin,gray] (-4,\i) -- (1,\i);
    }
\end{tikzpicture}
        }
        \caption{$t=5$}
  \end{subfigure} 
\caption{2 steps of the Turmite $RRL$.}\label{fig: ejem RRL}
\end{figure}
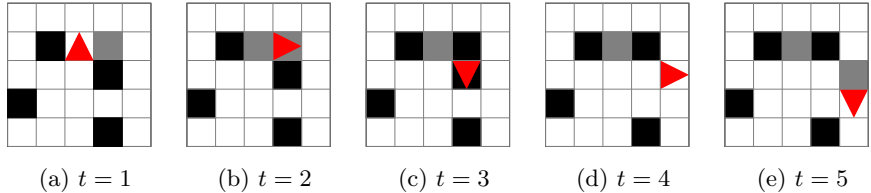

The Turmites whose turning direction is independent of the symbol are \emph{trivial}, since they will have a 4-periodic behavior independently of the initial configuration.
Turmites are a particular case of a Turing machine over a two-dimensional tape.
Non trivial Turmites have two important properties, proved in~\cite{GalProSutTro95a.mi}: they are \emph{reversible}, the previous state of each configuration can be uniquely determined, and their trajectories are always \emph{unbounded}, no periodic point can appear.

Langton's ant is the Turmite with rule $RL$.
It was introduced by Christopher Langton in~\cite{Langton86} together with other paradigmatic automata that emulate artificial life.
Figure~\ref{fig:Lang}) is the configuration obtained at iteration 11,200 when starting from the white con\-fi\-gu\-ra\-tion.
The repetitive pattern is produced after iteration 10,100, and it corresponds to a regular movement of period 104, known as the ``highway".
This particular asymptotic behavior has been experimentally observed over every finite initial configuration, but no one knows whether it will really occur in \emph{every} case, which is a long standing open problem.
In~\cite{GGM02} it was proved that universal computation can be performed with this ant. Since this is done by drawing a particular infinite initial configuration, it has no implications on the question about the highway, but it does justify in some sense the difficulties encountered when studying this system.

\begin{figure}[H]
  \centering
  \begin{subfigure}[b]{0.32\textwidth}
    \centering
        {
        \centering
        \includegraphics[width= 0.9\textwidth]{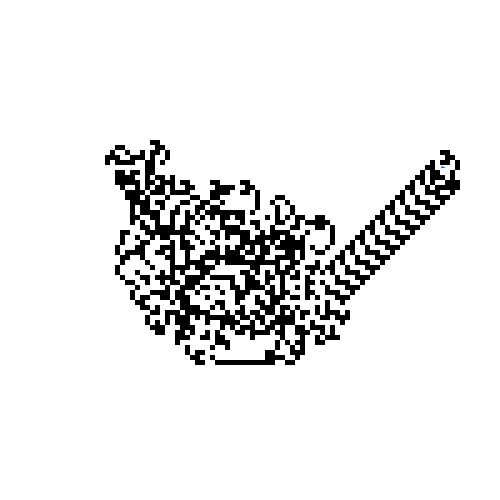}
        \caption{Rule RL, iteration 11200}\label{fig:Lang}
        }
  \end{subfigure}
  \begin{subfigure}[b]{0.32\textwidth}
    \centering
        {
        \centering
        \includegraphics[width= 0.9\textwidth]{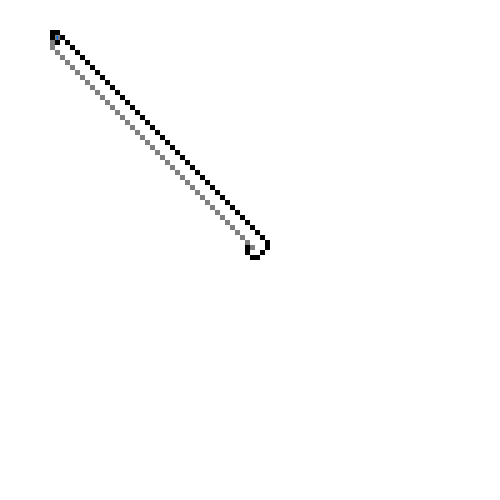}
        \caption{Rule RRL, iteration 786}\label{fig:RRL}
        }
  \end{subfigure}
  \begin{subfigure}[b]{0.32\textwidth}
    \centering
        {
        \centering
        \includegraphics[width= 0.9\textwidth]{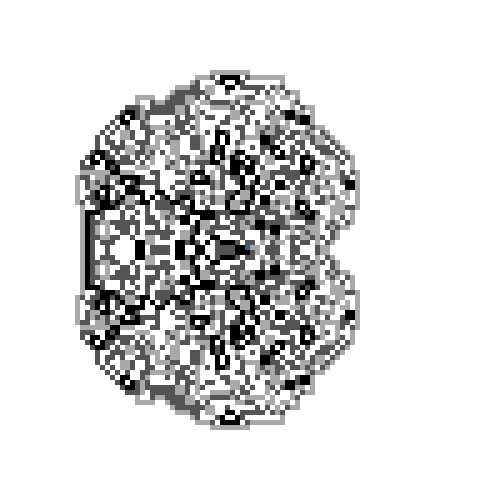}
        \caption{Rule RRLL, iteration 281498}
        }
  \end{subfigure}
  \caption{Turmite simulations: at the beginning all the cells are in state 0.}\label{fig:simulations}
\end{figure}

Other Turmites exhibit much more complex and diverse behaviors. Some of them build highways like Langton's ant, but with simpler patterns (as does rule $RRL$ in Figure~\ref{fig:RRL})), other grow squares instead, others have bilateral symmetry, etc.~\cite{GalProSutTro95a.mi}.


\section{Simulating universal computation \label{sec: sistema abstracto}}

We introduce a model of computation consisting in two basic elements, \emph{wires} and \emph{boxes}. 
The wires are used to transmit signals but can be only used once (see Figure~\ref{fig:wire}); they can be rotated and flexed.

The boxes are computational devices with a state ($0$ or $1$) that are connected to five (or three) wires (see Figure~\ref{fig:caja}).
When a box in state $0$ receives a signal from wire $1$, it flips from $0$ to $1$ and the signal exits through wire $2$.
The input wire 1 can only be used when the box is in state $0$.
Furthermore, a signal entering from wire $3$ will exit through wire $4$ if the current state is $0$, and through wire $5$ otherwise. 
The simplest such device is the \emph{input box}, where the wires 1 and 2 are left unused (Figure~\ref{fig: caja de entrada}).
These devices are represented in Figure~\ref{fig: elementos puertas cajas}.
They can be rotated but not reflected, that is, when enumerating the cables clockwise, we will always read 1-2-5-3-4.

\begin{figure}
  \centering
  \begin{subfigure}[b]{0.33\textwidth}
    \centering
	{
          \includegraphics[width=0.9\textwidth]{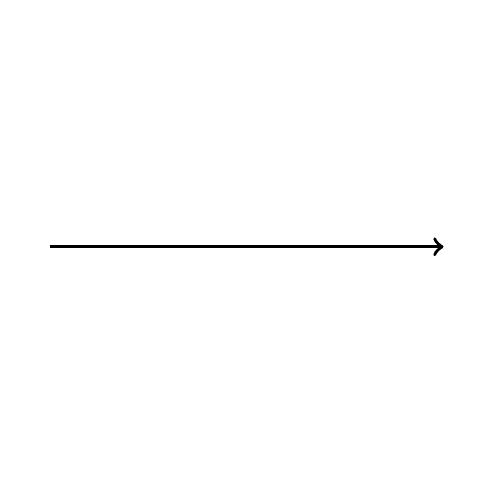}
	}
        \caption{Wire.} \label{fig:wire}               
  \end{subfigure}
  \begin{subfigure}[b]{0.33\textwidth}
    \centering
	{
	  \includegraphics[width=0.9\textwidth]{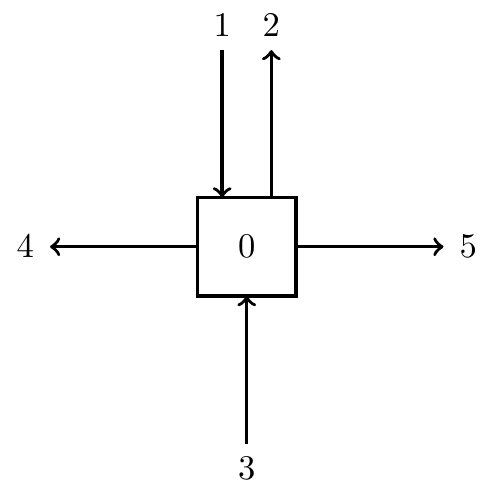}
	}
        \caption{Box in state $0$.} \label{fig:caja}     
  \end{subfigure}
  \begin{subfigure}[b]{0.33\textwidth}
    \centering
	{
	  \includegraphics[width=0.9\textwidth]{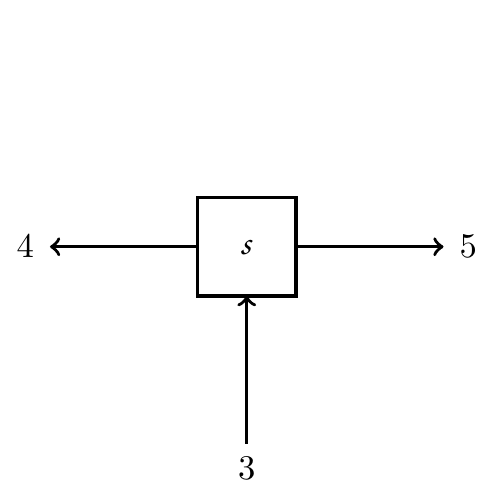}
	}
	\caption{Input box in state $s\in \{0,1\}$.}\label{fig: caja de entrada}
  \end{subfigure}
  \caption{Diagrams of the wire and the boxes used in the construction of logic gates.}\label{fig: elementos puertas cajas}
\end{figure}

Our system is intended to work in a two-dimensional space, which implies the need for a device to cross wires, and also another to join two wires into one; therefore, we introduce a \emph{cross} and a \emph{union}.
The \emph{union} has two input wires and one output wire; a signal coming from any of the input wires will exit through the output wire. 
The \emph{cross} allows two wires $1$ and $2$ to cross in such a way that their signals do not interact. This device works as long as wire $1$ is used {\em before} wire $2$, that is, wire $2$ can be used only if wire $1$ has been used before. 
Since this device can be rotated and reflected, wire 1 and 2 can be in any position.
Diagrams of these devices are shown in Figure \ref{fig: cruce}.

\begin{figure}
  \centering
  \begin{subfigure}[b]{0.33\textwidth}
    \centering
    \includegraphics[width=0.9\textwidth]{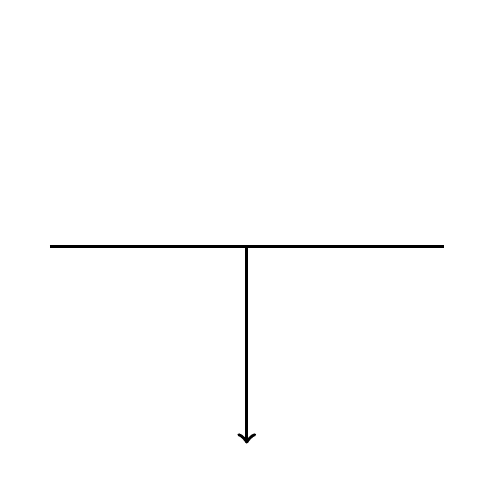}
    \caption{Union.}
  \end{subfigure}
  \begin{subfigure}[b]{0.33\textwidth}
    \centering
    \includegraphics[width=0.9\textwidth]{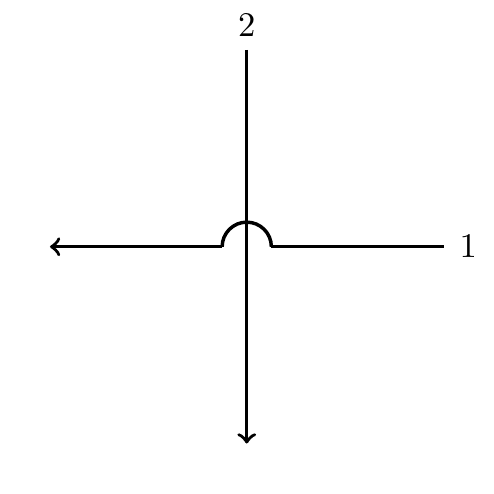}
    \caption{Cross.}\label{fig: cruce}
  \end{subfigure}
  \caption{Additional elements in the Wires and Boxes system. In the Cross, the wire $1$ will be always distinguished with a semicircle.}\label{fig: adicionales}
\end{figure}

From these elements we can build any logical gate, such as, for example, the NAND gate ($p\sim\!\!\wedge q \Leftrightarrow\ \sim(p\wedge q)$) or the NOT gate ($\sim$). 
Their construction is shown in Figure~\ref{fig: logic gates}.
Notice that the computation itself is due to the action of the boxes, and the signal coming from the left only triggers the computation by passing through the different boxes. 
From case to case analysis we can check that for any initial values of $p$ and $q$ in the upper boxes, the state of the lower box after the signal passed through will be $NAND(p,q)$ and $NOT(p)$, respectively.
Let us remark that since the boxes cannot be reflected, these gates cannot be reflected neither.
The NAND and the NOT gate can be only computed by a signal that enters by the left.

\begin{figure}
  \centering
  \begin{subfigure}[b]{0.64\textwidth}
    \centering
    \includegraphics[height=5cm]{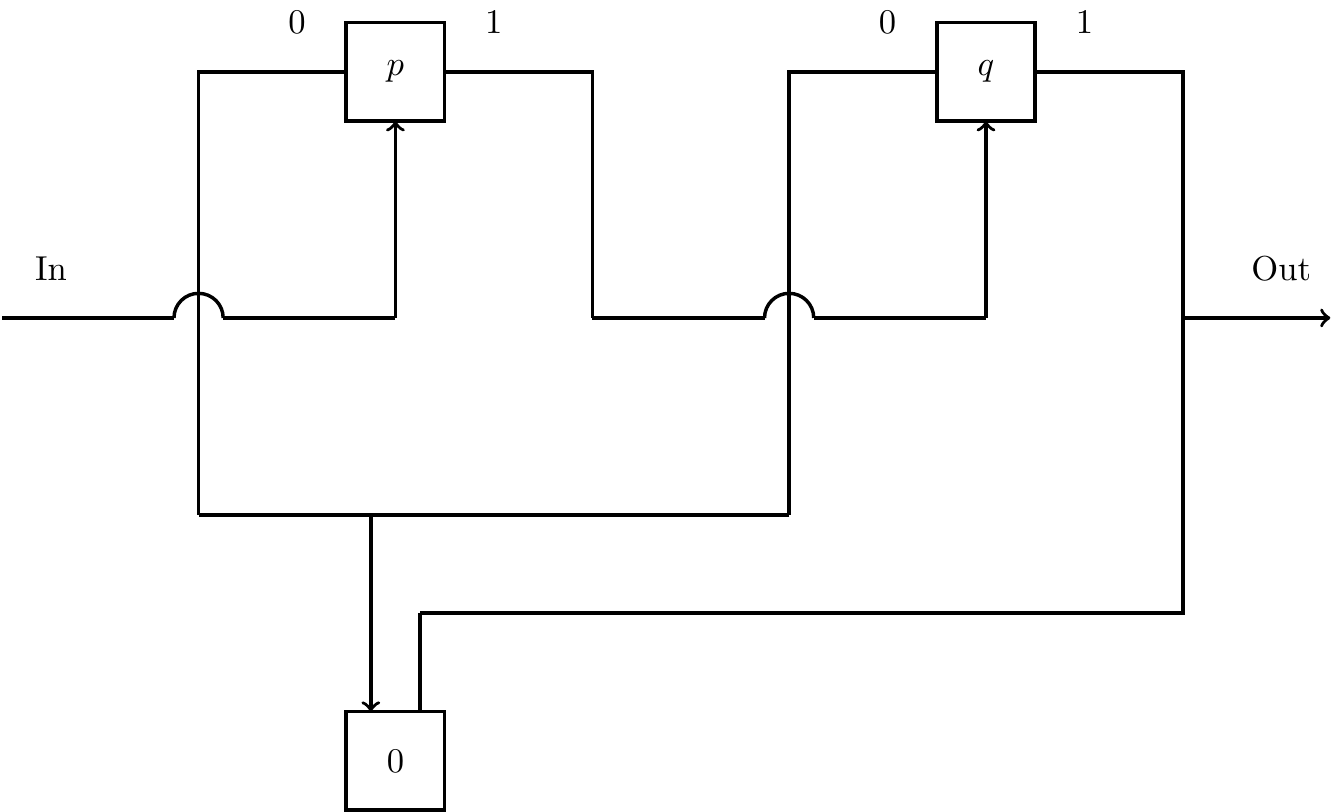}
    \caption{NAND gate.}                
  \end{subfigure}
  \begin{subfigure}[b]{0.34\textwidth}
    \centering    
    \includegraphics[height=5cm]{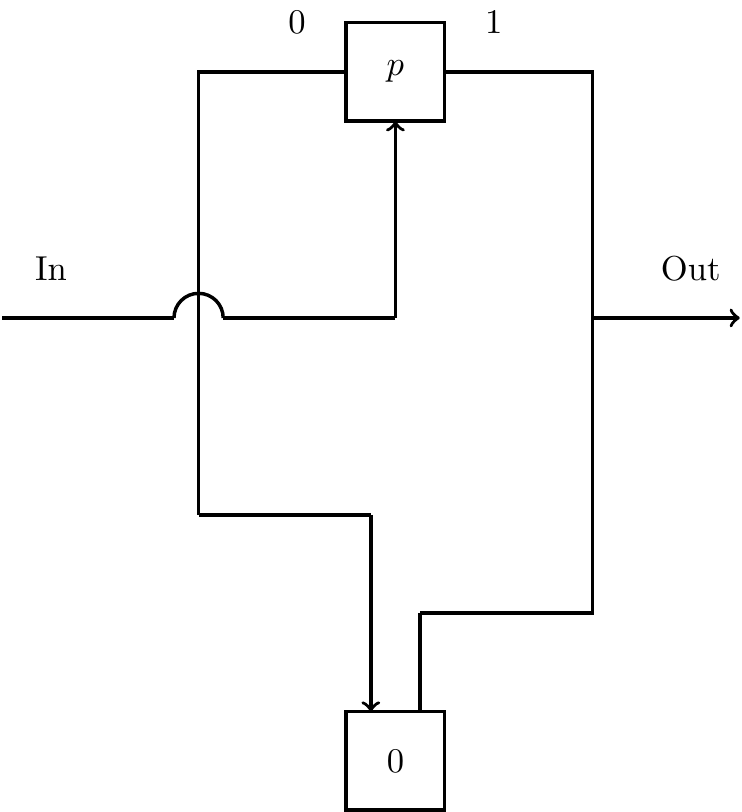}    
    \caption{NOT gate.}\label{fig: NOT gate}                
  \end{subfigure}
  \caption{Examples of logic gates built from wires and boxes.}\label{fig: logic gates}
\end{figure}

In order to perform compositions of logic gates, we should connect them by using the output (lower) box of a gate as input (upper) box of another gate. 
However the gates of Figure~\ref{fig: logic gates} only work with a signal going from left to right. 
To fix this issue, we design a special gate COPY that copies the state from one box to another while allowing the signal to travel from right to left. 
This gate can be also used to connect gates that are arbitrarily separated, and it can be also used to duplicate values. 
All these uses of the COPY gate are illustrated in Figure~\ref{fig: COPY gate}.
Like the NAND and the NOT gate, the COPY gate cannot be reflected, and it can be used only from right to left.

\begin{figure}
  \centering
  \begin{subfigure}[b]{0.34\textwidth}
    \centering
    \includegraphics[height=5cm]{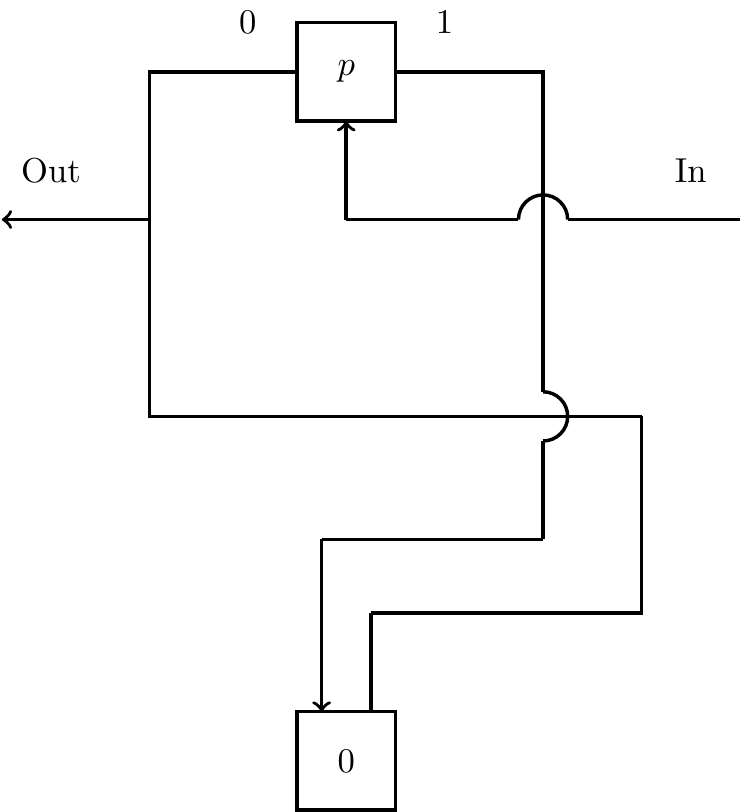}
    \caption{COPY gate.}
  \end{subfigure}
  \begin{subfigure}[b]{0.64\textwidth}
    \centering
    \includegraphics[height=5cm]{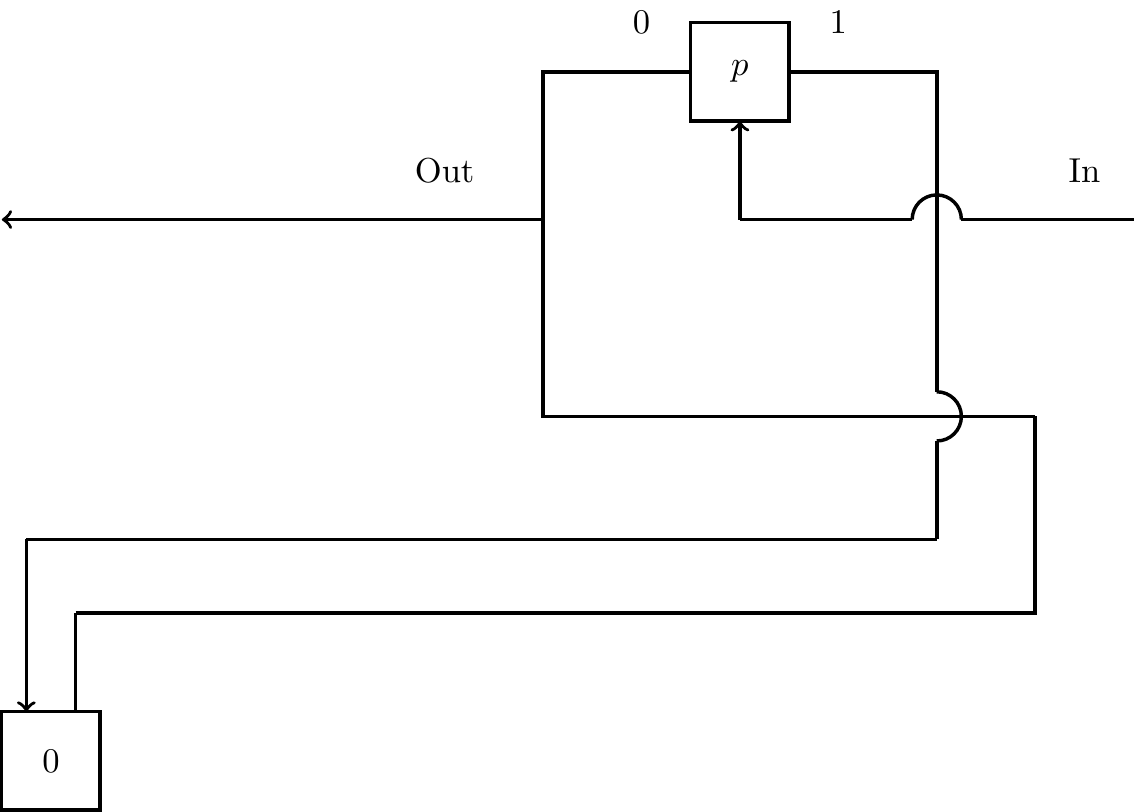}
    \caption{COPY gate moving a bit of information arbitrarily far.}\label{fig: copy-move}
  \end{subfigure}

  \vspace{15pt}

  \begin{subfigure}[b]{0.8\textwidth}
    \centering
    \includegraphics[height=5cm]{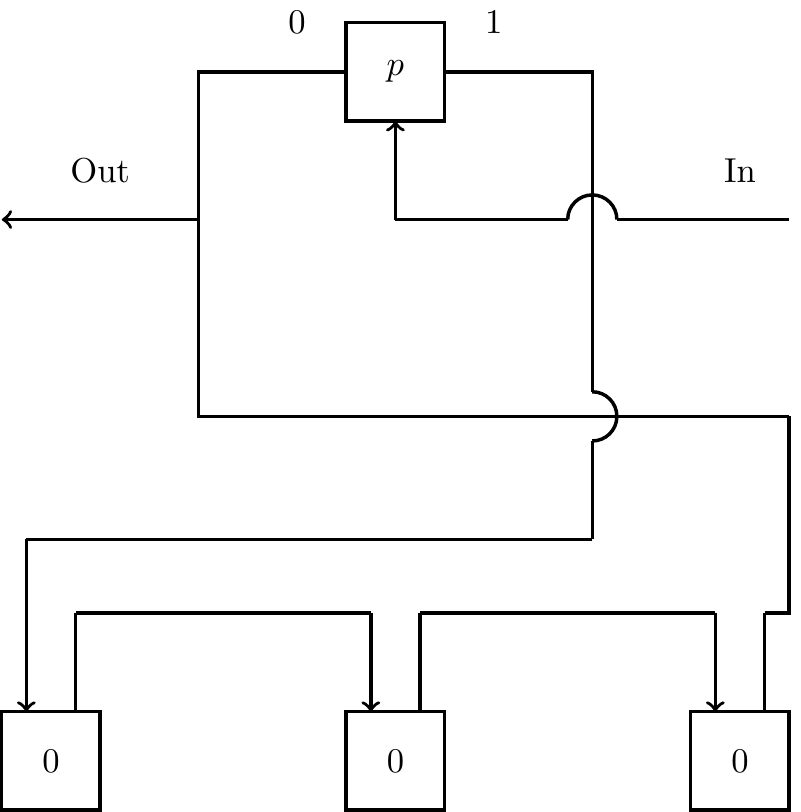}
    \caption{COPY gate triplicating a bit of information.}\label{fig: copy-duply}
  \end{subfigure}  
  
  \caption{Uses of the COPY gate.}\label{fig: COPY gate} 
\end{figure}

Our last device allows information to cross, and hence is named the CROSS gate.
A CROSS gate can be obtained from a XOR gate ($\veebar$), itself obtained by composition of NAND and NOT gates, as illustrated in Figure~\ref{fig: esq circ}.
  
\begin{figure}
  \centering
  \begin{subfigure}[t]{0.45\textwidth}
    \centering
    \includegraphics{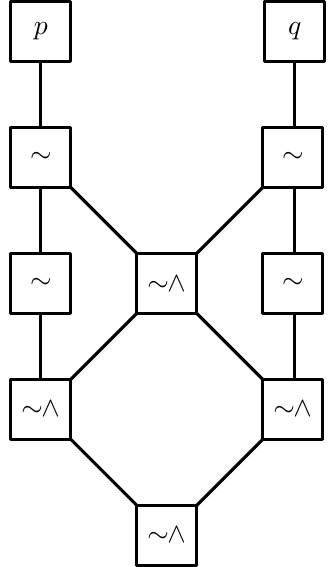}
    \caption{XOR gate built from NOT and NAND gates.}\label{fig: esq circ XOR} 
  \end{subfigure}
  \begin{subfigure}[t]{0.45\textwidth}
    \centering
    \includegraphics{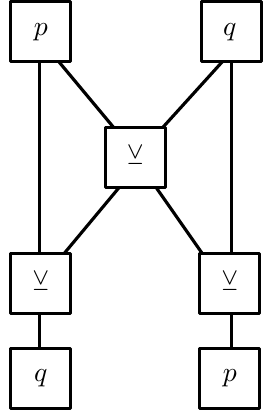}
    \caption{CROSS gate built from XOR gates.}\label{fig: esq circ CROSS} 
  \end{subfigure}    
  \caption{Construction of XOR and CROSS gates.}\label{fig: esq circ}
\end{figure}
  
From these constructions we obtain that the boxes and wires system can simulate any logical circuit.
Moreover, with an infinite diagram we can compute infinite logical circuits, layer by layer, and thus we can simulate the action of a cellular automaton over a finite initial configuration.
This, in turn, allows the simulation of an arbitrary Turing machine.

\begin{teor}
For any cellular automaton of radius $\frac{1}{2}$ with a quiescent state, there exists an infinite periodic wires and boxes diagram such that for any finite initial configuration it can be modified in a finite part in such a way that the diagram simulates the cellular automaton on the given configuration.
\end{teor}
\begin{proof}

Without loss of generality we can assume that the cellular automaton has $2^N$ states, has local function $f$, and that the state $0$ with binary representation $(0,\ldots,0)$ is quiescent. In the simulation, a finite configuration ${}^\omega 0 a_0 a_1 \ldots a_n 0^\omega$ will be represented by a row of the form 
\[
{}^\omega 0 
\mathbf{01} \left \lfloor{a_0}\right \rfloor \left \lfloor{a_0}\right \rfloor 
\mathbf{00}  \left \lfloor{a_1}\right \rfloor \left \lfloor{a_1}\right \rfloor 
\mathbf{00} \ldots  
\mathbf{00}  \left \lfloor{a_n}\right \rfloor \left \lfloor{a_n}\right \rfloor 
\mathbf{10}
0^\omega
\]
where $\left \lfloor{a_i}\right \rfloor$ is the binary representation of $a_i$.
Thus, each ``cell'' of the automaton is duplicated and translated into binary.
The pairs of digits in bold will serve a special function: a 1 on the left will travel to the right, indicating the end of the (growing) finite configuration, and a 1 on the right will travel to the left, indicating its beginning. 
Moreover, each row will be duplicated through the use of a layer of {\em copy} gates, where the computation proceeds from right to left.
To simulate the iteration of the cellular automaton we will use a layer of copies of a circuit $C$, computed from left to right, whose details we give below.
The general structure of the simulating configuration is shown in Figure~\ref{fig: ACSim}, where $a'_i=f(a_{i-1},a_i)$ and $a''_i=f(a'_{i-1},a'_i)$ (except for the first and last, where an argument is replaced by the quiescent 0).

\begin{figure}
  \centering
  \includegraphics[width=\textwidth]{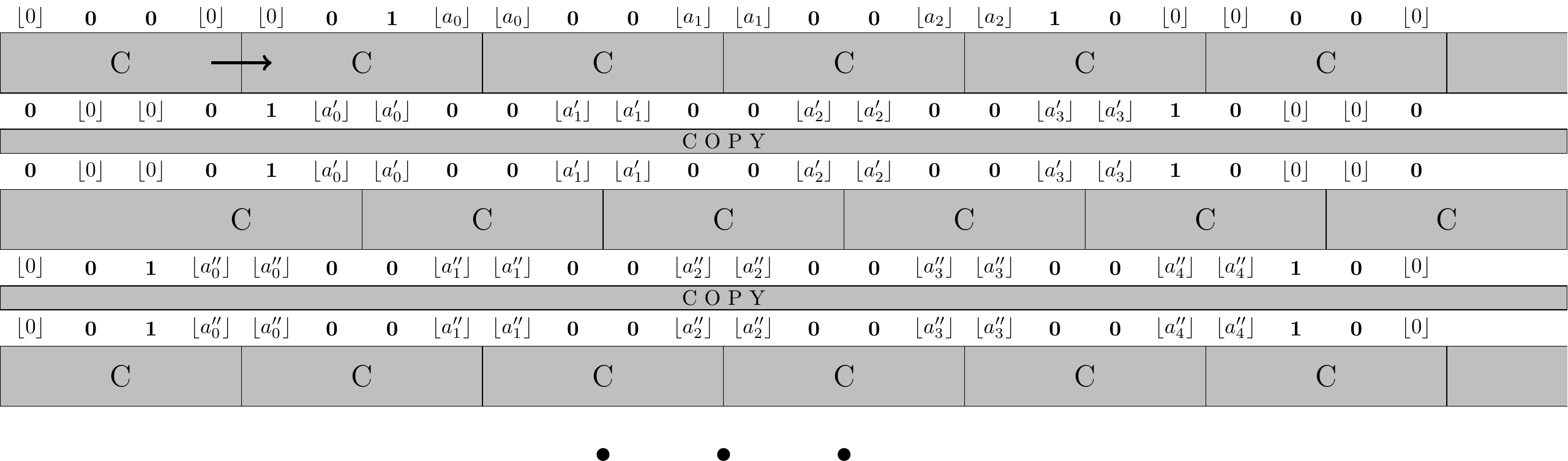}
  \caption{Simulation of a radius 1/2 cellular automaton with quiescent state, through the use of boxes and cables.
    Layers of copy gates alternate with layers computing the local function.
    The Turmite starts on the bold arrow.}\label{fig: ACSim}
\end{figure}

The purpose of circuit $C$ is to compute the boolean vector function
\[
 \varphi(x^L_0,x^L_1,\ldots,x^L_{N-1},s_0,s_1,x^R_0,x^R_1,\ldots,x^R_{N-1}) \;=\;
 \begin{cases} 
 (s_1,f(x^L,x^R),f(x^L,x^R),s_0), & \mbox{if } s_1=0\\ 
 (s_1,f(0,x^R),f(0,x^R),0), & \mbox{if } s_1=1
 \end{cases}
\]
where the local function $f$ of the automaton has been adapted to the binary representation.
As shown by the constructions in the previous section, such a circuit can always be constructed.

The construction so far would simulate the iteration of the cellular automaton, if we could compute each infinite row after the other.
Since we want the computation to limit itself to the (increasing) non-quiescent part, we must stop calculating the $C$ circuits at some point and turn to the copying row, and then stop that too and restart the computation of $C$ at a lower level.
In order to keep the background configuration periodic, this must be achieved by a modification of the circuits which allows the flow of computation to obey the signals $s_0$ and $s_1$.

First, we modify $C$ by asking that the first and last rows of circuits inside it perform a NOT on its inputs (resp., its outputs).
Thus, the inner workings of $C$ will actually compute $\sim{\varphi(\sim{x^L},\sim{s_0},\sim{s_1},\sim{x^R})}$.
This modification allows us to intervene $C$ without having to know the precise structure of the computation of $\varphi$ itself.
The intervention itself is depicted in Figure~\ref{fig: ACDetalle}, and involves the bottom right corner of $C$, the middle-top part of the $C$ below it, and the copying gates in between.
We know that the bottom right corner of $C$ is a NOT that will write into its output only if $s_0=1$ and $s_1=0$.
Hence, we reroute the cable after this 1 has been written, and instead of continuing to the next $C$, we move downwards and join the copying gates, finishing the computation of the current row.
On the other hand, we modify the gate that copies the value of $s_1$ into the $C$ circuit below: if the value is 1, we interrupt the copying row and move into $C$, joining the row of NOTs with which it begins.
The function $\varphi$ ignores the left part of its input when $s_1=1$, in this way the circuit $C$ is correctly computed in this case, even if the Turmite did not enter through the input wire of $C$.

\begin{figure}
  \centering
  \includegraphics[width=0.5\textwidth]{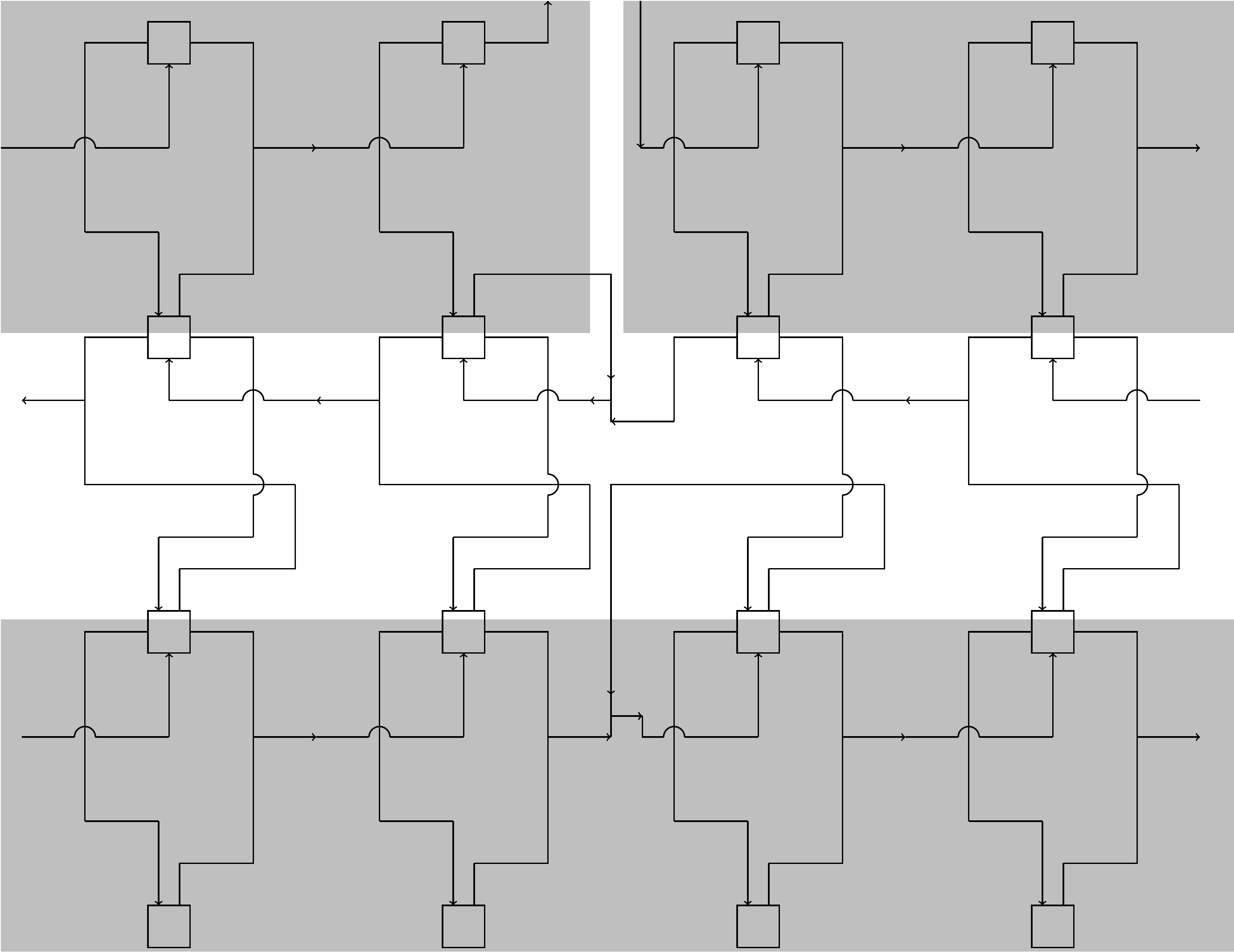}
  \caption{Changes to the NOT and COPY gates around the bottom corners of the $C$ circuit.
    Three different instances of $C$ are highlighted in gray.
    The NOT of the bottom right corner of the first is modified, as well as the COPY gate below the bottom right corner of the second.}\label{fig: ACDetalle} 
\end{figure}  

\end{proof}

In particular, if we manage to simulate an arbitrary boxes and cables circuit inside a system (in this case, the trajectory of a Turmite), then we are able to simulate universal computation.

%
%
%
 
\section{Turing-universality of Turmites}\label{sec:Turmit}

In this section we show that any nontrivial Turmite can build through its trajectory all the elements of the cables and boxes system.
The proof is based on the fact that any nontrivial Turmite rule $w$ has at least two colors $r,l$ such that $w(r)=R$, $w(r\oplus_n 1)=L$, $w(l)=L$ and $w(l\oplus_n 1)=R$.
In all the following diagrams we will represent $r$ as white, and $l$ as black.
Cells that are not represented can have an arbitrary color.

\paragraph{Cables.} Cables going from left to right correspond to a line of $l$ cells under a line of $r$ cells, and can be followed by a Turmite initially oriented to the east as shown in Figure~\ref{fig: way ex}. Since the Turmite rule is invariant by rotation, right-left, up-down and down-up cables can be constructed in the same way. These cables are connected through a ``L" shape with the appropriate corner color, as shown in Figures~\ref{fig: L ex} and \ref{fig: L ex2}.
  
\begin{figure}[H]
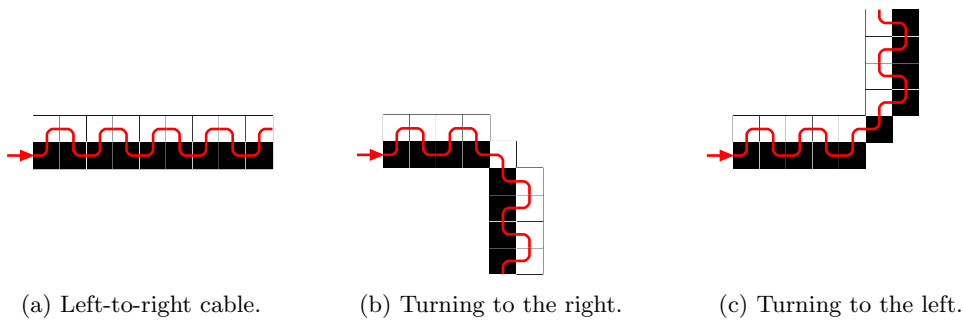

  \centering 
  \begin{subfigure}[t]{0.3\textwidth}
    \centering
    \includegraphics{img/cable_a}
    \caption{Left-to-right cable.}\label{fig: way ex} 
  \end{subfigure}
  \begin{subfigure}[t]{0.3\textwidth}\centering 
    \includegraphics{img/cable_b}
    \caption{Turning to the right.}\label{fig: L ex} 
  \end{subfigure}
  \begin{subfigure}[t]{0.3\textwidth}\centering 
    \includegraphics{img/cable_c}
    \caption{Turning to the left.}\label{fig: L ex2} 
  \end{subfigure}  
  \caption{Building single-use cables in nontrivial Turmites. The red line represents the Turmite trajectory.}
\end{figure}

\paragraph{Boxes.} Boxes are built according to Figure~\ref{fig: esq caja}. The state of the box is defined by the common color of the two cells in the yellow rectangle (white for 1, black for 0).

A signal entering by 1 follows the red line and exits by 2 after modifying the state of the box, as expected (corresponding to Figure~\ref{fig:caja}).
A signal entering through 3 follows the green line and exits through 4 or 5 depending on the state of the box. 
  
\begin{figure}[H]
  \centering
  \includegraphics[width=4cm]{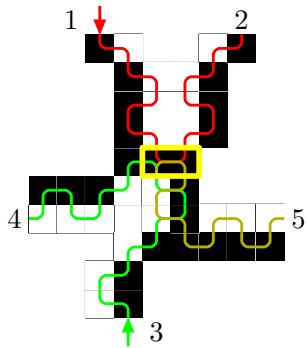}
  \caption{The Box. The internal state is recorded in the yellow rectangle, which is initially black (0). It can be modified by the upper input cable (red). The result is ``read'' from bellow, by the lower cable (green).} 
\label{fig: esq caja}
\end{figure}  

\paragraph{Cross and union.} Finally, the Cross and Union are built as indicated in Figures~\ref{fig:union}, \ref{fig:cruceA} and~\ref{fig:cruceB}. 
All of these Figures can be rotated but not reflected, since reflection changes the turning senses. 
This is not a limitation in the case of Cables and Unions, nor in the case of the Box, which is only needed in the way that it was designed, but we do need the Cross in a reflected form, which is why we include its two versions.
From these constructions we finally obtain our main result.

\begin{figure}[H]
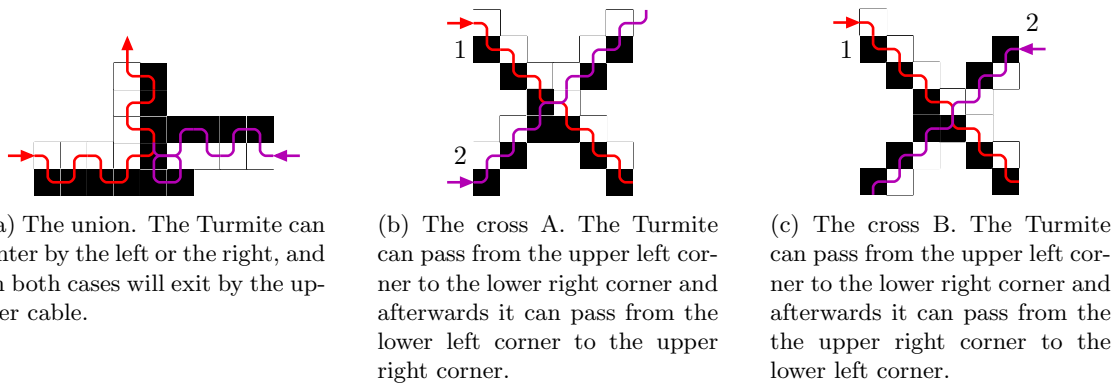

\centering
  \begin{subfigure}[t]{0.30\textwidth}
  \centering
  \includegraphics{img/union.mps}
  \caption{The union. The Turmite can enter by the left or the right, and in both cases will exit by the upper cable.}\label{fig:union} 
  \end{subfigure}  
  \hspace{.03\textwidth}
  \begin{subfigure}[t]{0.30\textwidth}
  \centering
  \includegraphics{img/cruceA.mps}
  \caption{The cross A. The Turmite can pass from the upper left corner to the lower right corner and afterwards it can pass from the lower left corner to the upper right corner.}\label{fig:cruceA}
  \end{subfigure} 
  \hspace{.03\textwidth} 
  \begin{subfigure}[t]{0.30\textwidth}
  \centering
  \includegraphics{img/cruceB.mps}
  \caption{The cross B. The Turmite can pass from the upper left corner to the lower right corner and afterwards it can pass from the the upper right corner to the lower left corner.}\label{fig:cruceB}
  \end{subfigure}  
\caption{Building cable union and crossover with nontrivial Turmites.}
\end{figure}    
 
  \begin{teor}\label{teo: uni RRL}
Any nontrivial Turmite can simulate the Cables and Boxes system, and therefore any 1D Turing machine and cellular automaton over finite initial configurations.
  \end{teor}

\section{\textbf{P}-completeness}\label{sec:Phardness}

The complexity class \textbf{P} is the class of decision problems that can be solved in polynomial time. It contains the class \textbf{NC} of decision problems that can be solved by a circuit of polynomial size and polylogarithmic depth; intuitively, \textbf{NC} corresponds to problems that can be efficiently parallelized. It is believed that $\textbf{NC}\neq\textbf{P}$.
A problem $A$ is \textbf{P}-complete if every problem in $P$ reduces to $A$ through an algorithm that uses logarithmic space. 
If $\textbf{NC}\neq\textbf{P}$, then a \textbf{P}-complete problem cannot be in \textbf{NC}.

Predicting the behavior of a computable system - given an initial state and a time $t$, predict the state or some property of the system after $t$ steps - can be done by simple simulation in time $O(t)$. In particular, the problem is in \textbf{P} if we see $t$ as the size of the input - i.e. $t$ is written in unary and the initial state is described using $O(t)$ bits. If the prediction of a system is \textbf{P}-complete, the system is seen as unpredictable in the sense that this naive algorithm is optimal; in other words, there is no ``shortcut'' to simulate the system, even using massively parallel computation.

In this section we consider a particular prediction problem. The VISIT($w$) problem for a fixed Turmite rule $w$ is defined as follows:

\begin{description}
\item[VISIT($w$).] Given a finite initial configuration $x$, a cell $(X,Y)\in\Z^2$ and a time $t\in\NN$, decide whether starting from $x$ the Turmite with rule $w$ passes by the cell $(X,Y)$ before time $t$.
\end{description}

We show the \textbf{P}-completeness of this problem. The proof proceeds by reduction from the Topologically Ordered Circuit Value problem (TCV), which is \textbf{P}-complete (see for example~\cite{Greenlaw:1995}). This proof is inspired by hints given in~\cite{Greenlaw:1995} and by the proof of \textbf{P}-completeness of Planar Circuit Value problem (PCV)~\cite{Goldschlager1977} through reduction to the Circuit Value problem (CV).

A logic circuit is a labeled directed acyclic graph $(V,E,G)$, where the vertices are labeled by $\{\textbf{AND}, \textbf{OR}, \textbf{NOT}, \textbf{I}\}$, and only one node has out degree 0: the \emph{output} node.
Moreover, vertices with labels in $\{\textbf{AND}, \textbf{OR}\}$ have in-degree 2 and vertices with label \textbf{NOT} have in-degree 1.
The input nodes are those with label \textbf{I}, and have in-degree 0.
In such a graph, when boolean values are assigned to the input nodes, the value of other nodes can be computed, in particular the value of the output node is the value of the circuit.

A logic circuit $(V,E,G)$ is said to be \emph{topologically ordered} if vertices are represented by numbers in such a way that 
for every $(u,v)\in E$ it holds that $u<v$.

\begin{description}
\item[TCV.] Given a topologically ordered logic circuit $L=(V,E,G)$ and a binary input $b$; decide whether the value of the circuit $L$ on $b$ is $0$ or $1$.
\end{description}

We put special care into proving that the reduction is in fact log-space. 
First let us define formally the input encoding of each problem.
For VISIT($w$), we use the following encoding.
\[(x_1,y_1,c_1)(x_2,y_2,c_2)...(x_m,y_m,c_m)\#(x_0,y_0)\#d\#(X,Y)\#t\]
Where every cell of the initial configuration is $0$ except for the cells $(x_i,y_i)\in\Z^2$ which have color $c_i$, the initial Turmite position is $(x_0,y_0)$, and its initial direction is $d$, the potentially visited cell is $(X,Y)$ and $t$ is the time bound, all coded in binary.
As per the usual convention in $\textbf{P}$-completeness proofs in similar systems, we consider time complexity with respect to $t$.

Given this input, VISIT($w$) is clearly decidable in a time polynomial in $t$: build a $2t\times 2t$ central square of the initial configuration,  simulate the first $t$ steps of the trajectory of the Turmite and check whether the cell $(X,Y)$ was visited. This requires $\mathcal{O}(t^2)$ operations.

For TCV, the topologically ordered logic circuit is encoded as follows.
\[(1,g_1,l_1,r_1)(2,g_2,l_2,r_2)...(n,g_n,l_n,r_n)\#(i_1,v_1)(i_2,v_2)...(i_k,v_k)\]
where $g_i$ is the gate type of vertex $i$ and $l_i,r_i$ are the left and right input vertices of gate $i$ if any.
The pair $(i_j,v_j)$ represents an input vertex $i_j$ with its input value $v_j$.
The vertex $n$ is assumed to be the \emph{output} gate.
Using this encoding, the input length of a circuit is $\mathcal{O}(n\log(n))$.

\begin{prop}
VISIT($w$) is P-complete for every non trivial Turmite rule $w$.
\end{prop}

\begin{proof}
Given a topologically ordered logic circuit with gates $\{g_1,...,g_n\}$, we first build an equivalent circuit whose elements are positioned in a $n\times n$ grid. For each $1\leq i\leq n$, a gate $g_i$ (built using the boxes and wires system seen in Section~\ref{sec: sistema abstracto}) is placed at coordinate $(i,i)$ in the lattice. Each gate has two vertical cables bringing its left and right input from above, and one horizontal cable carrying its output towards the right. If an input is not used, the cable is left unconnected.

At each coordinate $(k,i)$ with $k<i$, the cable carrying the output of $g_k$ and the two cables carrying the input of $g_i$ meet. We distinguish three cases:

\begin{enumerate}
\item If $l_i = k$ then the horizontal cable is connected to the first vertical cable and crosses the second vertical cable;
\item Symmetrically if $r_i = k$;
\item Otherwise, the horizontal cable crosses both vertical cables.
\end{enumerate}

Note that, since the original circuit is topologically ordered, all cable unions take place in the upper right half of the grid.
Therefore we can omit sending cables to the left of the gates.
Furthermore, we can ``activate'' each gate in numerical order, since each gate's input only depends on the output of gates with a lower number.
This consideration will prove useful considering the sequential nature of the computation in the Turmite's trajectory.

These crosses and connections are performed using the devices shown in Figure~\ref{fig:Tpe}, one device for each possible case. Notice that an horizontal cable that is connected to a vertical cable also continues towards the right.

Furthermore, since the Turmite travels from left to right, we have to perform a ``carriage return'' when the Turmite reaches the end of a line to bring it to the beginning of the next line, for which we use the COPY gates. The time needed to travel through the circuit is at most $t=n^2T$, where $T$ is the maximal time to travel through one of the devices. In order to guarantee that the Turmit will not visit the output cell after exiting the circuit, a wire of length $t$ is added at the end of the circuit.

\paragraph{Example.} Figure~\ref{fig:TCVXOR} shows an example of this construction for a small circuit.
The values of the input gate included while drawing the input gate.

  \begin{figure}[htpb]
    \centering
    \begin{subfigure}[b]{0.3\textwidth}
      \centering
      \includegraphics[height=5cm]{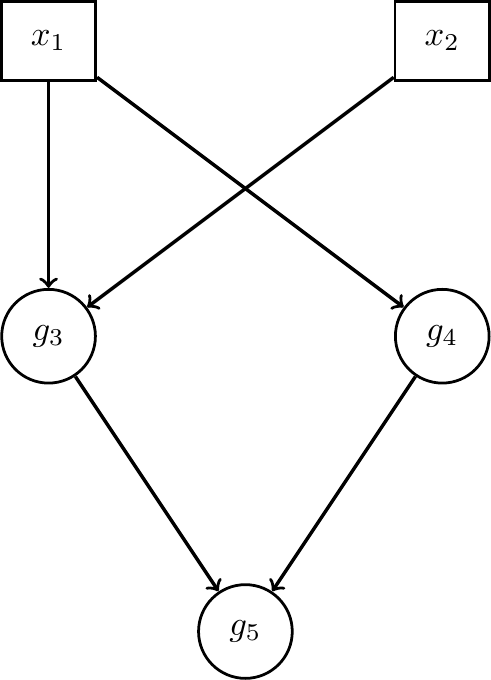}
      \caption{Original circuit.}              
    \end{subfigure}%
    \begin{subfigure}[b]{0.35\textwidth}
      \centering
      \includegraphics[height=5cm]{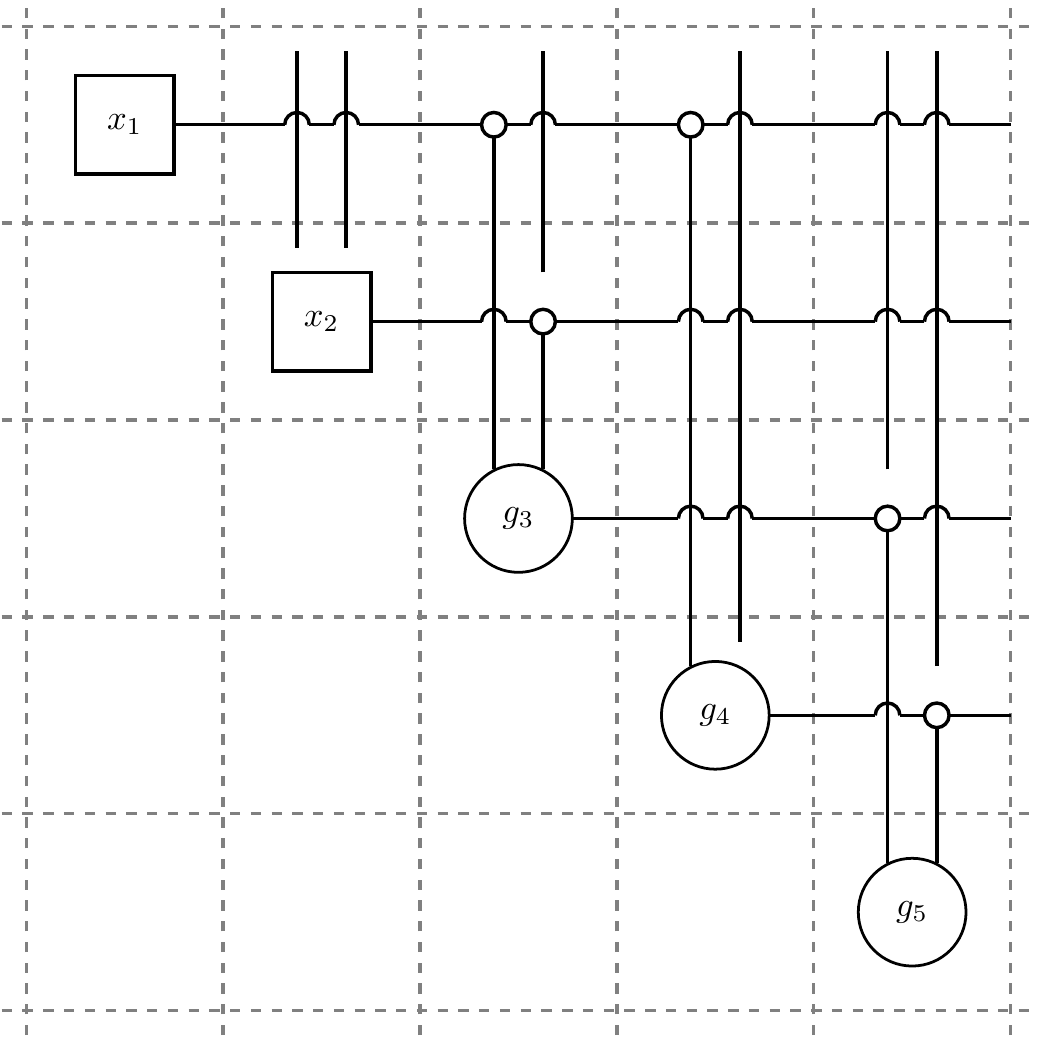}
      \caption{Resulting circuit.}                
    \end{subfigure}
    \begin{subfigure}[b]{0.34\textwidth}
      \centering
      \includegraphics[height=5cm]{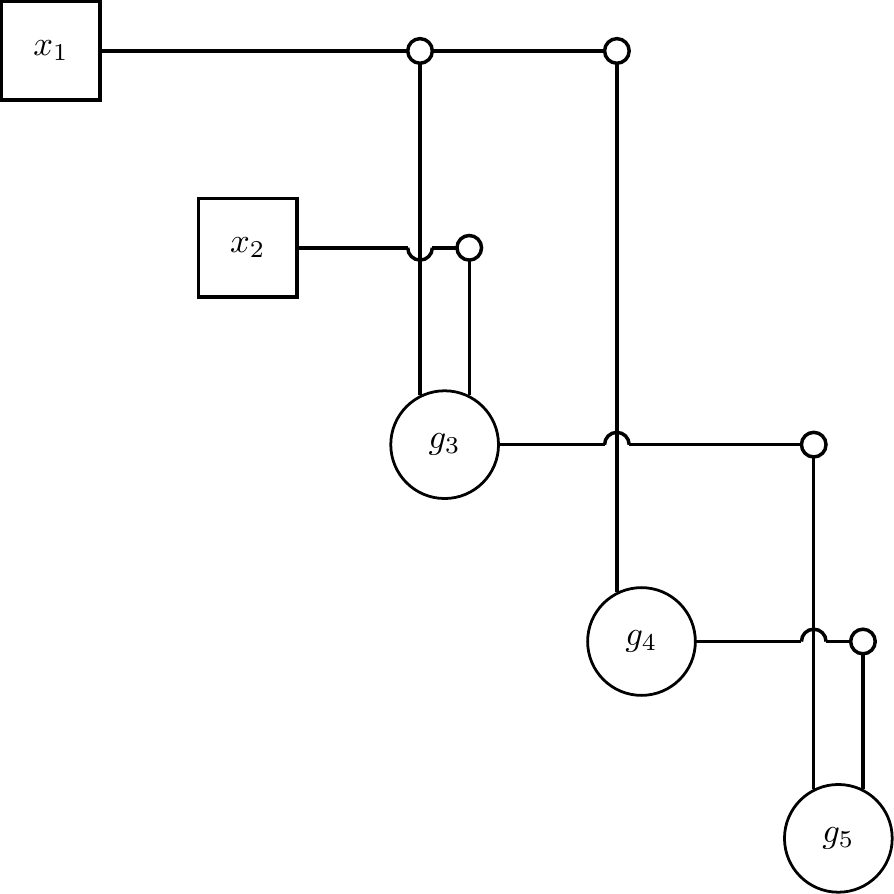}
      \caption{Removing unused cables.}                
    \end{subfigure}
    \caption{Example of a conversion of an given circuit to a circuit embedded in a grid.}\label{fig:TCVXOR} 
  \end{figure}

The output of the circuit is computed by the configuration shown in Figure~\ref{fig: ejem imp alg}, where the Turmit trajectory is represented by a red line.

\begin{figure}[htpb]
  \centering
  \begin{subfigure}[t]{0.315\textwidth}
    \centering
    \includegraphics[width=\textwidth]{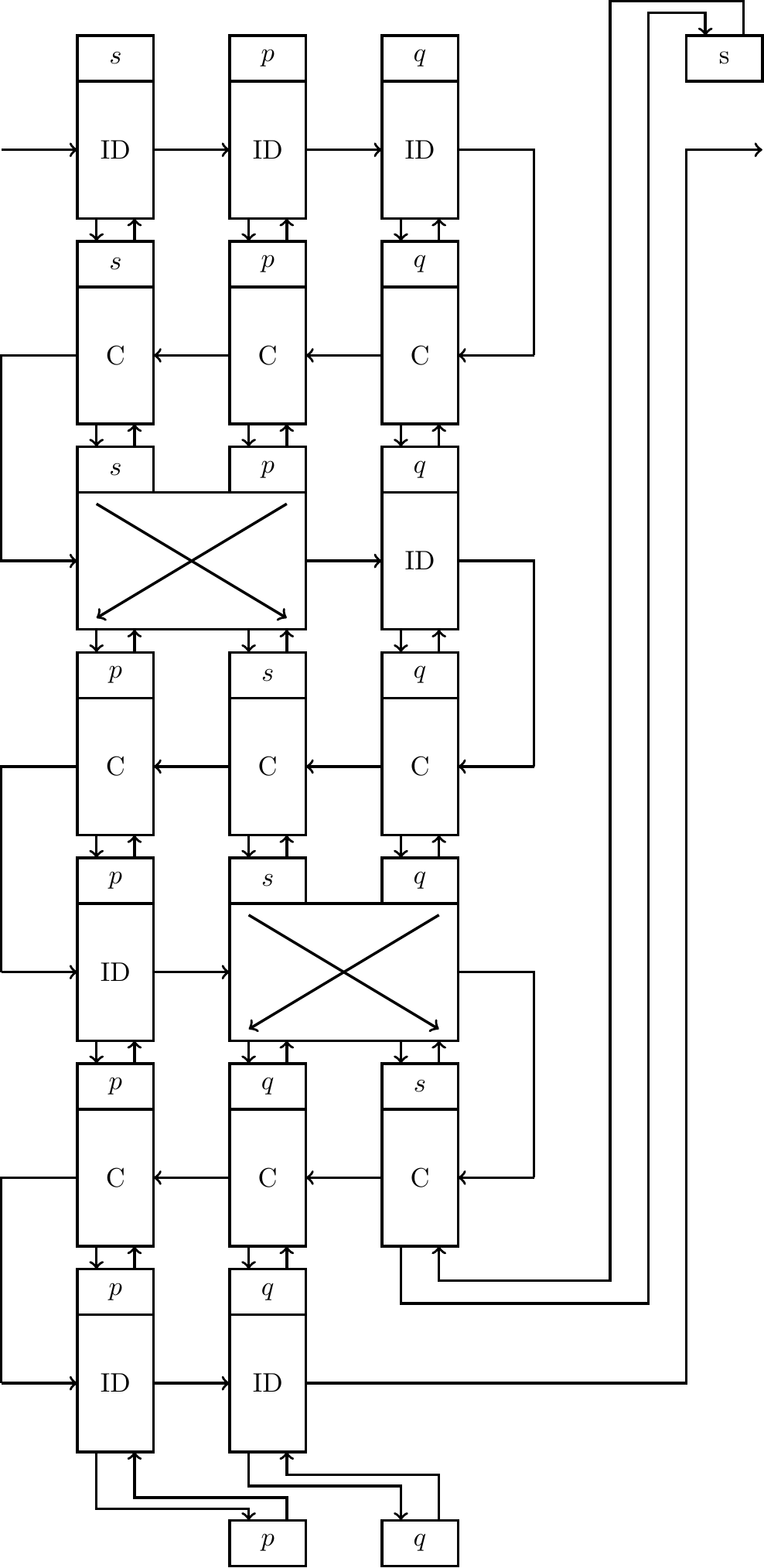}
    \caption{An horizontal rightwards cable crosses two vertical downwards cables.}\label{fig: red cross}                
  \end{subfigure}%
  \hspace{0.01\textwidth}
  \begin{subfigure}[t]{0.315\textwidth}
    \centering
    \includegraphics[width=\textwidth]{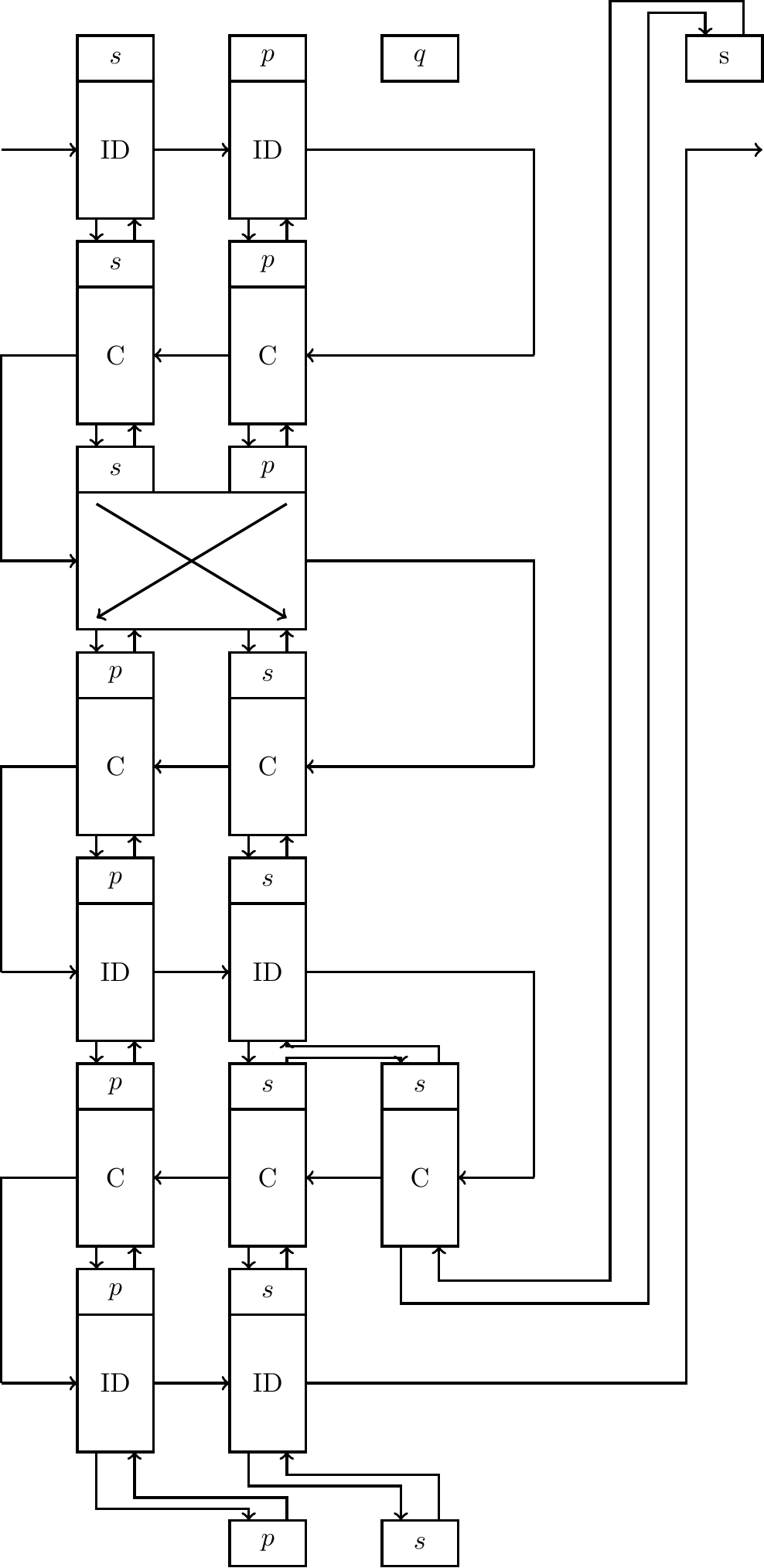}
    \caption{An horizontal rightwards cable crosses two vertical downwards cables and connects with the second one.}\label{fig: red join1}                
  \end{subfigure}%
  \hspace{0.01\textwidth}
  \begin{subfigure}[t]{0.315\textwidth}
    \centering
    \includegraphics[width=\textwidth]{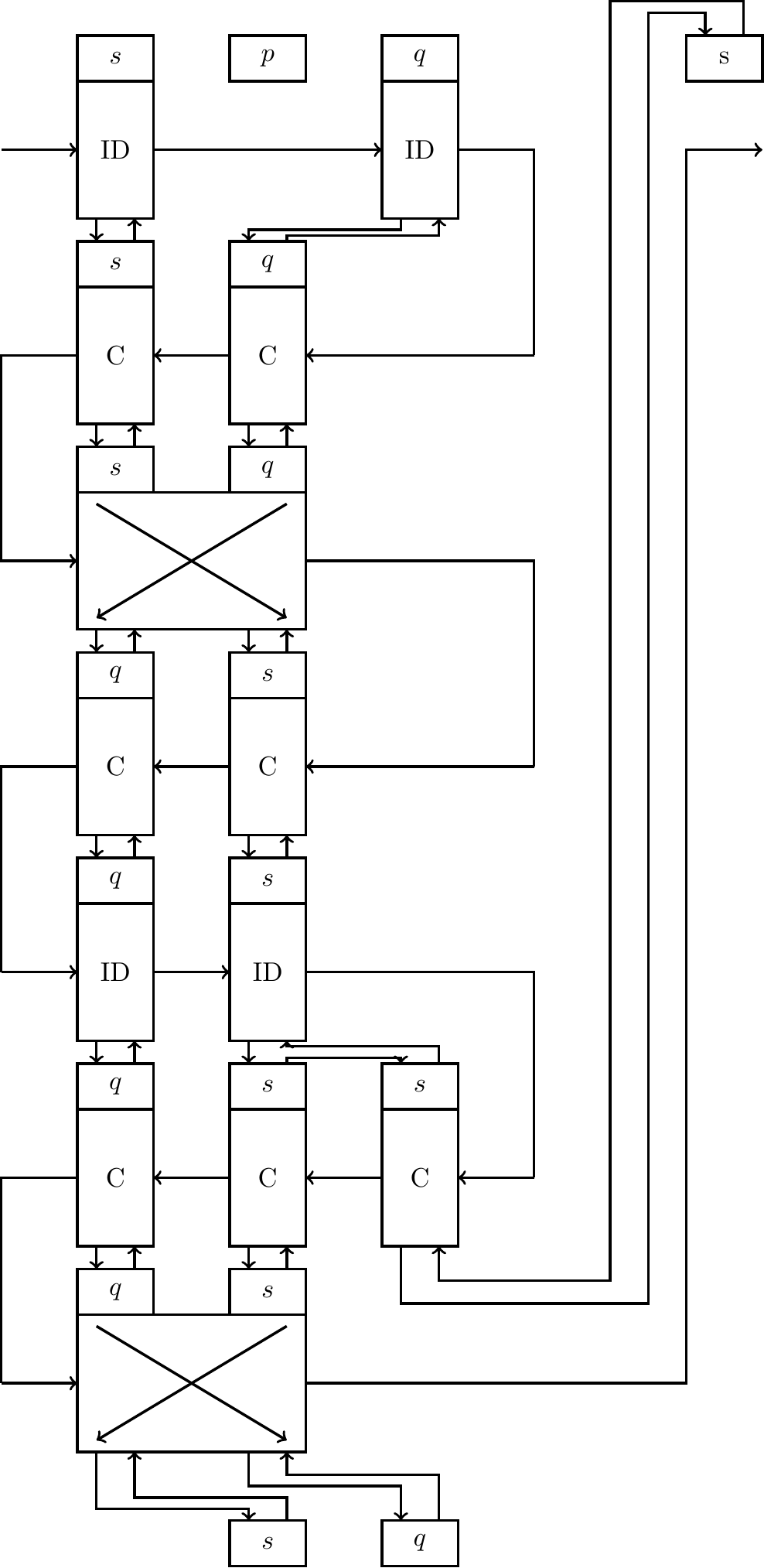}
    \caption{An horizontal rightwards cable crosses two vertical downwards cables and connects with the first one.}\label{fig: red join2}                
  \end{subfigure}
  \caption{Gates performing the union and crossing of cables used in the reduction of TCV to VISIT. The $ID$ gate is a shortcut for two NOT gates connected through a COPY gate. C is a shortcut for COPY. The crossed gate is the planar cross.} \label{fig:Tpe}
\end{figure}

  \begin{figure}[htbp]
    \centering
    \begin{subfigure}[t]{0.27\textwidth}
      \centering
      \includegraphics[height=2.5cm]{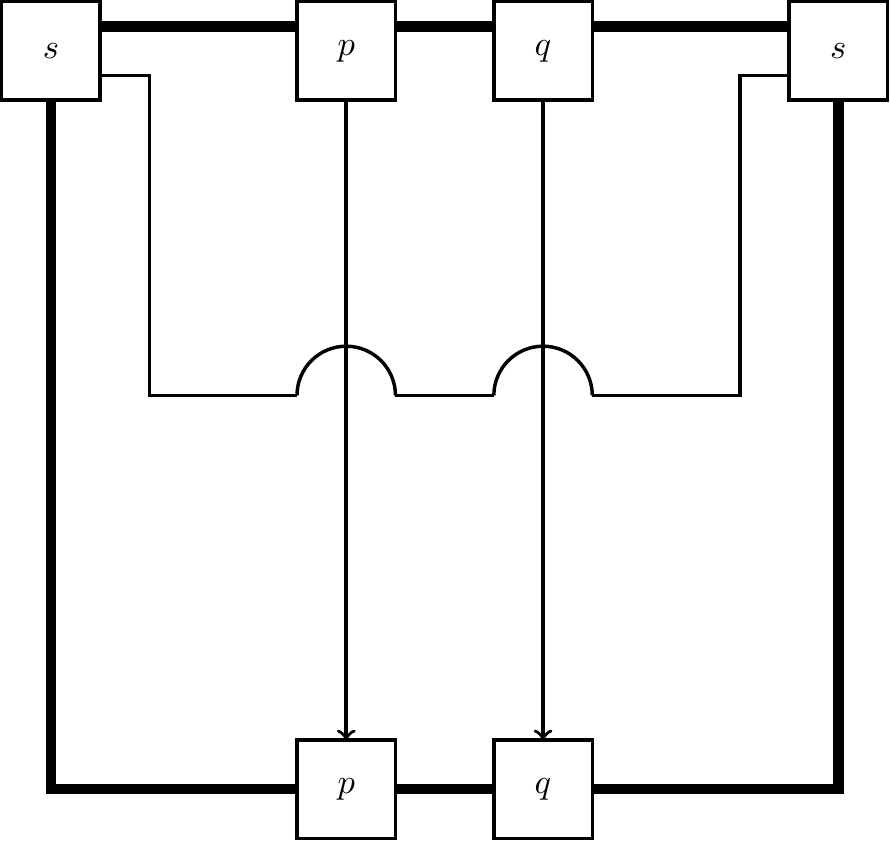}
      \caption{Figure~\ref{fig: red cross}), simplified.}
    \end{subfigure}
    \hspace{0.06\textwidth}
    \begin{subfigure}[t]{0.27\textwidth}
      \centering
      \includegraphics[height=2.5cm]{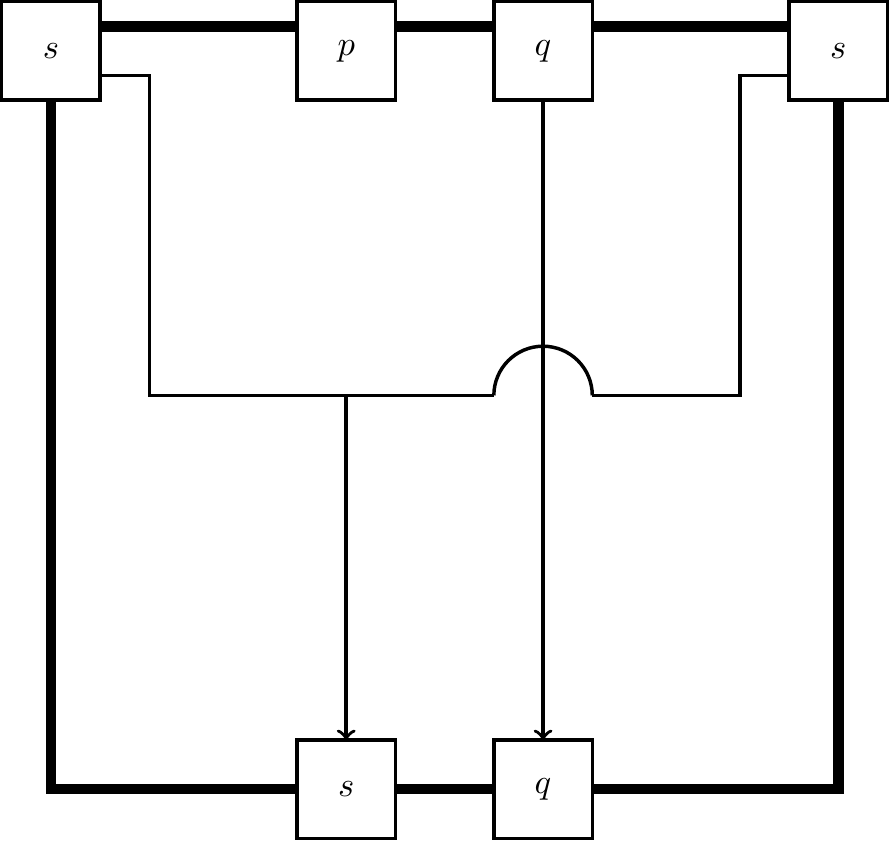}
      \caption{Figure~\ref{fig: red join1}), simplified.}  
    \end{subfigure}
    \hspace{0.06\textwidth} 
    \begin{subfigure}[t]{0.27\textwidth}
      \centering
      \includegraphics[height=2.5cm]{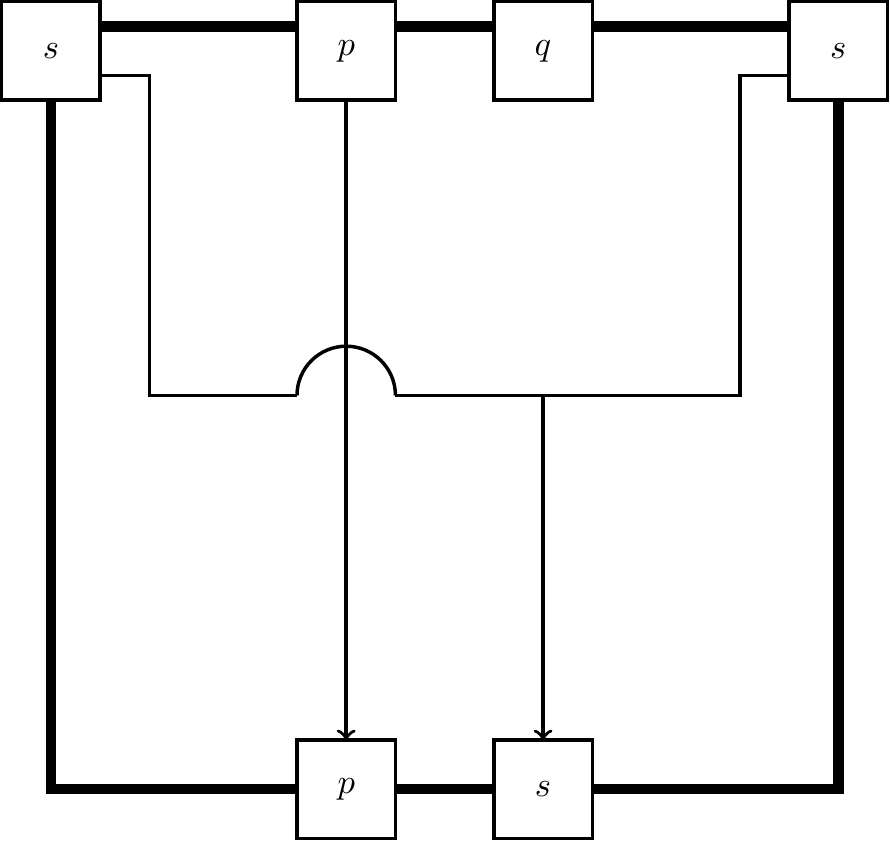}
      \caption{Figure~\ref{fig: red join2}), simplified.}  
    \end{subfigure}
  \end{figure}

 \begin{figure}[htbp]
  \centering
  \includegraphics[width=\textwidth]{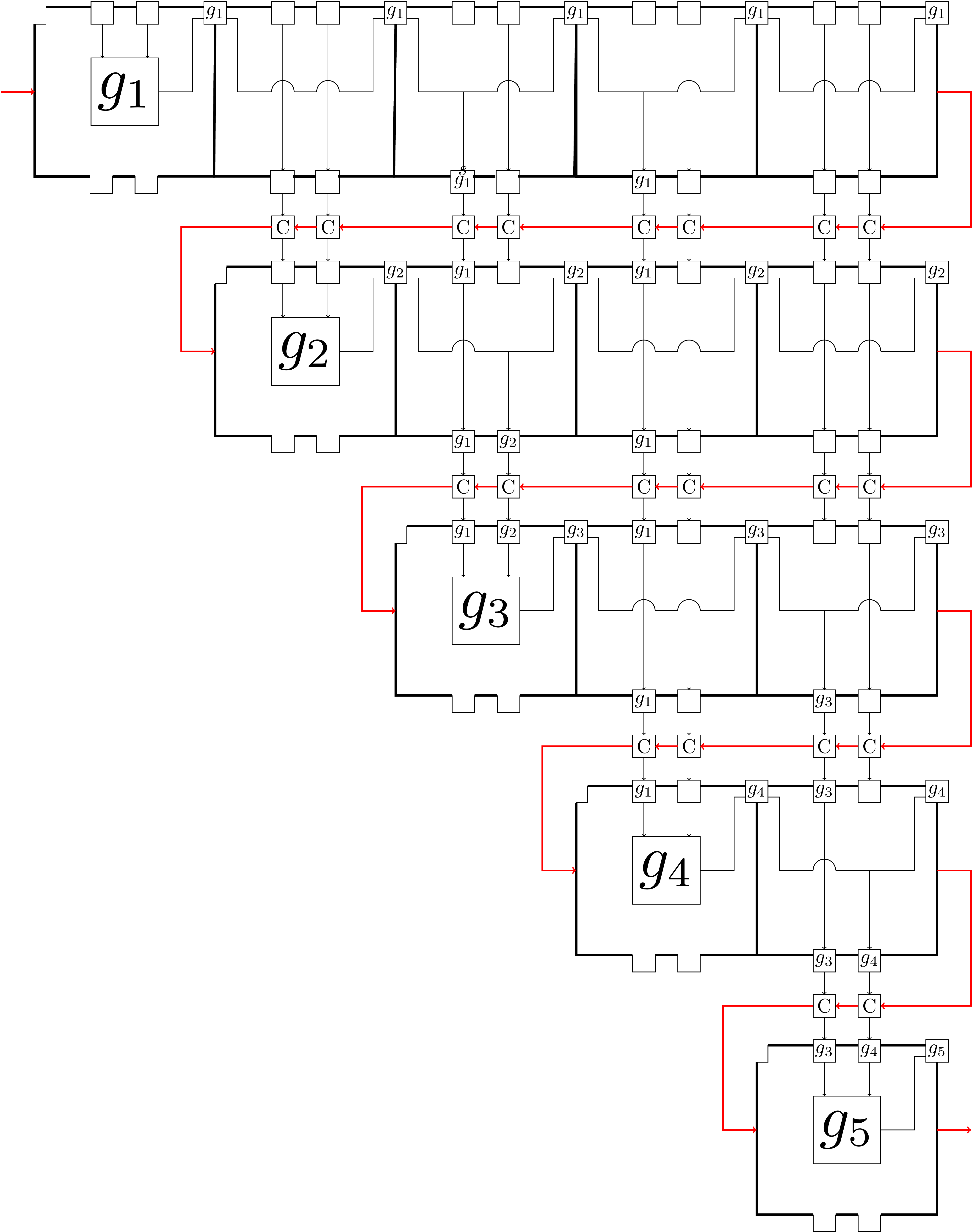}
  \caption{A Turmite configuration simulating the logic circuit. Gate $G_i$ corresponds to the $i$-th gate of the original circuit, and $g_i$ is the output of the gate $G_i$. C is the COPY gate. The red line corresponds to the Turmite trajectory.}\label{fig: ejem imp alg}
 \end{figure}  

To complete the proof of TCV $\leq_{NC}$ VISIT($w$), we show that the reduction can be performed in logarithmic space and polynomial time with the input encodings defined above. The detailed algorithm is as follows:
  
  \begin{algorithm}[H]
  \caption{Reduction of TCV to VISIT($w$).} \label{alg; reduc}
  \begin{algorithmic} 
  \REQUIRE $(1,g_1,l_1,r_1)(2,g_2,l_2,r_2)...(n,g_n,l_n,r_n)\#(i_1,v_1)(i_2,v_2)...(i_k,v_k)$	
    \STATE $row  \leftarrow 1$
    \STATE $col  \leftarrow 1$    
    \STATE $Out  \leftarrow \espacio$
    \FOR{$row=1,...,n$}
      \STATE $Out  \leftarrow Out + \text{GATE$(row)$}$
      \FOR{$col=row+1,...,n$}
      	\STATE $P  \leftarrow $ FIND-PRED$(row,col)$
        \STATE $Out  \leftarrow Out+ \text{CROSS$(P,row,col)$}$
    \ENDFOR
    \STATE $Out  \leftarrow Out+\text{COPY-ROW$(row)$}$	   
    \ENDFOR
    \STATE $Out \leftarrow Out+ $END-ROUTINE$(row)$
    \ENSURE $Out$ 
  \end{algorithmic}
  \end{algorithm}
  
To understand how this algorithm works, first divide ${\Z^2}$ in \emph{tiles} whose side length is the maximum side length between the devices of Figure~\ref{fig:Tpe}. Variables $r$ and $c$ correspond to the current row and column, respectively.

The algorithm goes through the $n$ rows of the tile grid. The column $r$ of row $r$ contains the $r$-th gate.
In other columns $c$ $(r<c)$ we check whether $(r,c)$ is an edge of the circuit through the subroutine FIND-PRED.
Depending on the result, we write the appropriate crosses or unions using the subroutine CROSS.

When a row is finished, a row of COPY gates is added underneath so that the Turmite performs a carriage return to the tile $(r+1,r+1)$ to begin the computation of the next row. This is performed by the subroutine COPY-ROW. Finally we add to the output the coordinates of the last box and the time when it should have been visited, using the subroutine END-ROUTINE.

Note that Out, the output of the algorithm, is a write-only variable. Therefore this algorithm can be performed using a transducer. By construction, the output cell of the output gate is visited if and only if the output value of the circuit is $1$ (remember that the output cell of a gate is never visited if the output is $0$ - see Figure~\ref{fig: esq caja}).

\begin{description}  
   \item[GATE$(r)$:] This subroutine reads the input for the type of the $r$-th gate, then returns the description of this gate to be written in the tile $(r,r)$. If the gate is an input gate, the description also includes the input value of the gate.

   Its complexity is $\mathcal{O}(n)$ in time (to scan the input) and $\mathcal{O}(\log{n})$ in space (to write $r$).
   \item[FIND-PRED$(r,c)$:] This subroutine finds the value of $l_c$ and $r_c$ to check whether gate $g_r$ is a predecessor of $g_c$ in the circuit. It outputs $1$ if $l_c=r$, $2$ if $r_c=r$ and $0$ otherwise. Its complexity is $\mathcal{O}(n)$ in time (to scan the input) and $\mathcal{O}(\log{n})$ in space.
   \item[CROSS$(P, r,c)$:] This subroutine simply outputs the device of the corresponding type as seen in Figure~\ref{fig:Tpe}, using constant time and space.
   \item[COPY-ROW$(r)$:] This subroutine writes a row of COPY gates below the tiles $(r,r+1), \dots, (r,n)$. It has complexity $O(n)$ in time and $O(\log(n))$ in space (to keep the value of the counter $r+1 \to n$)
   \item[END-ROUTINE$(f)$: ] This subroutine adds the last wire of length $t = n^2T$, and the additional information needed by the encoding: $(x_1,x_2) = (n,n)$ the coordinates of the output gate, $(0,0)$ the initial coordinates of the Turmite, $d$ the initial direction of the Turmite, and $t = n^2T$. Its complexity is $O(\log(n))$ in time and space.
\end{description}

\paragraph{Time complexity.}
The subroutines FIND-PRED and CROSS are called up to $n^2$ times; subroutines GATE and COPY-ROW are called $n$ times; and subroutine END-ROUTINE is called once. All these subroutines have time complexity $O(n)$. Therefore the time complexity of the reduction algorithm is $\mathcal{O}(n^3)$.

\paragraph{Space complexity.}
The algorithm uses up to $4\log{n}$ memory bits outside of the subroutines, and each subroutine has space complexity $\mathcal{O}(\log(n))$. Therefore the space complexity of the algorithm is $\mathcal{O}(\log{n})$.

This concludes the proof.
\end{proof}

\section{Discussion}\label{sec:discusion}

The results presented here improve and generalize those found in \cite{GGM01} and \cite{GGM02}, where the universality and P-completeness of Langton's ant had been shown. In the case of Turing-universality, the construction has been simplified through the use of radius 1/2 cellular automata, and has been modified in order to have a periodic background (except for the finite perturbation which encodes the input). The latter puts the result in line with the standards of the literature, akin to the well known proof of universality for the elementary cellular automaton 110 in~\cite{Cook04a}. In the case of P-completeness, the proof is much more rigorous and verifies that the reduction is indeed log-space.

Moreover, the results apply to all non-trivial Turmites, despite their very different behaviors. The key to achieving this generality is that all the constructions require the Turmite to visit any given cell at most two times; since non-trivial rules have at least both a $R\rightarrow L$ and a $L\rightarrow R$ transitions, the states can be chosen to comply with the requirements of the different gadgets. A case that we have not considered here is that of ``semitrivial'' rules, characterized by an infinite word of the form $L^nR^\omega$, for some $n$. However, their dynamics appear to be heavily limited: $R$ cells are needed if we want the Turmite to move at all, but afterwards they turn into a wall that freezes any further movement. Hence, we do not expect universality nor P-completeness from them.

A possible further contribution of this article is the isolation of the abstract logic of the construction, in the form of the cables and boxes system, from the technical realization of this system in the particular case of Turmites. Since both the universality and P-hardness depend only on the abstract construction, we hope that the same setting might be useful for similar results in other systems (e.g., other 2D Turing machines, or more generally, other systems of simple agents interacting with some environment). To this end, we tried to keep the devices at a bare minimum.

\bibliographystyle{alpha}
\bibliography{biblio}

\begin{thebibliography}{GHR95b}

\bibitem[Ban71]{AITR-233}
Edwin~Roger Banks.
\newblock Information processing and transmission in cellular automata.
\newblock Technical Report AITR-233, MIT Artificial Intelligence Laboratory,
  1971.

\bibitem[Coo04]{Cook04a}
Matthew Cook.
\newblock Universality in elementary cellular automata.
\newblock {\em Complex Systems}, 15(1):1--40, 2004.

\bibitem[Dew89]{Dewd89}
Alexander~K. Dewdney.
\newblock Computer recreations: Two-dimensional {Turing} machines and
  {Tur-mites} make tracks on a plane.
\newblock {\em Scientific American}, pages 124--127, September 1989.

\bibitem[DR99]{DurRok99}
Bruno Durand and Zsuzsanna R{\'{o}}ka.
\newblock The game of life: Universality revisited.
\newblock In M.~Delorme and J.~Mazoyer, editors, {\em Cellular Automata},
  volume 460 of {\em Mathematics and Its Applications}, pages 51--74. Springer
  Netherlands, 1999.

\bibitem[GG07]{GGchaitin}
Anah{\'{\i}} Gajardo and Eric Goles.
\newblock Circuit universality of two dimensional cellular automata: a review.
\newblock In Cristian Calude, editor, {\em Randomness and complexity, from
  Leibniz to Chaitin}, chapter~7, pages 131--152. World Scientific Pub., 2007.

\bibitem[GGM01]{GGM01}
Anah{\'\i} Gajardo, Eric Goles, and Andr{\'e}s Moreira.
\newblock Generalized {Langton's} ant: dynamical behavior and complexity.
\newblock In Afonso Ferreira and Horst Reichel, editors, {\em STACS}, volume
  2010 of {\em Lecture Notes in Computer Science}, pages 259--270. Springer,
  2001.

\bibitem[GHR95a]{Greenlaw95}
Raymond Greenlaw, H.~James Hoover, and Walter~L. Ruzzo.
\newblock {\em Limits to Parallel Computation: P-completeness Theory}.
\newblock Oxford University Press, Inc., New York, NY, USA, 1995.

\bibitem[GHR95b]{Greenlaw:1995}
Raymond Greenlaw, Howard Hoover, and Walter Ruzzo.
\newblock {\em Limits to Parallel Computation: P-completeness Theory}.
\newblock Oxford University Press, Inc., New York, NY, USA, 1995.

\bibitem[GMG02]{GGM02}
Anah{\'{\i}} Gajardo, Andr{\'e}s Moreira, and Eric Goles.
\newblock Complexity of {L}angton's ant.
\newblock {\em Discrete Appl. Math.}, 117(1-3):41--50, 2002.

\bibitem[Gol77]{Goldschlager1977}
Leslie~M. Goldschlager.
\newblock The monotone and planar circuit value problems are log space complete
  for {P}.
\newblock {\em SIGACT News}, 9(2):25--29, July 1977.

\bibitem[GPST95]{GalProSutTro95a.mi}
David Gale, Jim Propp, Scott Sutherland, and Serge Troubetzkoy.
\newblock {F}urther {T}ravels with {M}y {A}nt.
\newblock {\em Math. Intelligencer}, 17:48--56, 1995.

\bibitem[Lan86]{Langton86}
Christopher Langton.
\newblock Studying artificial life with cellular automata.
\newblock {\em Physica D}, 22:120--149, 1986.

\bibitem[NW06]{NearyWoods}
Turlough Neary and Damien Woods.
\newblock P-completeness of cellular automaton rule 110.
\newblock In {\em International Colloquium on Automata, Languages, and
  Programming}, pages 132--143. Springer, 2006.

\end{thebibliography}

\end{document}